\documentclass[sigconf]{acmart}
\setcopyright{none}

\usepackage[utf8]{inputenc}
\usepackage{graphicx}
\usepackage{latexsym}
\usepackage{amsmath}
\usepackage{amssymb}
\usepackage{url}
\usepackage{float}
\usepackage{tikz}
\usepackage{fancyvrb}
\usepackage{listings}
\usepackage{xcolor}
\usepackage{color}
\usepackage{soul}
\usepackage{multirow}
\usepackage{xspace}
\usepackage{multicol}
\usepackage{subfig}
\usepackage{booktabs}
\usepackage{caption}

\usepackage[ruled, linesnumbered]{algorithm2e}
\usepackage{algpseudocode}
\usepackage{amsthm}
\usepackage{mathrsfs}

\DeclareMathOperator*{\minimize}{minimize}
\newcommand{\cc}{\textnormal{\textbf{CatEdgeClus}}}
\newcommand{\mwc}{\textnormal{\textbf{MultiwayCut}}}
\newcommand{\cut}{\textnormal{\textbf{cut}}}

\newcommand{\obj}{Categorical Edge Clustering}

\newcommand{\clustering}{Y}
\newcommand{\given}{\;\vert\;}

\newcommand{\xhdr}[1]{\vspace{0.5mm}\noindent{\textbf{#1.}}\hspace{0.5mm}}

\copyrightyear{}
\acmYear{}
\setcopyright{none}
\acmConference{}
\acmBooktitle{}
\acmPrice{}
\acmDOI{}
\acmISBN{}

\PassOptionsToPackage{hyphens}{url}\usepackage{hyperref} 
\definecolor{mylinkcolor}{RGB}{0,0,140}
\hypersetup{colorlinks,allcolors=mylinkcolor,citecolor=mylinkcolor}
\usepackage[capitalize]{cleveref}
\crefname{figure}{Figure}{Figures}
\crefname{algocf}{Algorithm}{Algorithms}
\Crefname{algocf}{Algorithm}{Algorithms}

\title{Clustering in graphs and hypergraphs\\ with categorical edge labels}

\begin{abstract}
Modern graph or network datasets often contain rich structure that goes beyond
simple pairwise connections between nodes. This calls for complex
representations that can capture, for instance, edges of different types as well
as so-called ``higher-order interactions'' that involve more than two nodes at a
time. However, we have fewer rigorous methods that can provide insight from such
representations. Here, we develop a computational framework for the problem of
clustering hypergraphs with categorical edge labels --- or different interaction
types --- where clusters corresponds to groups of nodes that frequently
participate in the same type of interaction.

Our methodology is based on a combinatorial objective function that is related
to correlation clustering on graphs but enables the design of much more
efficient algorithms that also seamlessly generalize to hypergraphs. When there
are only two label types, our objective can be optimized in polynomial time,
using an algorithm based on minimum cuts. Minimizing our objective becomes
NP-hard with more than two label types, but we develop fast approximation
algorithms based on linear programming relaxations that have theoretical cluster
quality guarantees. We demonstrate the efficacy of our algorithms and the scope
of the model through problems in edge-label community detection, clustering with
temporal data, and exploratory data analysis.
\end{abstract}

\author{Ilya Amburg}
\affiliation{%
  \institution{Cornell University}
}
\email{ia244@cornell.edu}

\author{Nate Veldt}
\affiliation{%
  \institution{Cornell University}                                                                                                                                                           
}
\email{nveldt@cornell.edu}

\author{Austin R.~Benson}
\affiliation{%
  \institution{Cornell University}
}
\email{arb@cs.cornell.edu}

\begin{document}

\maketitle

\section{Introduction}
Representing data as a graph or network appears in numerous application domains,
including, for example, social network analysis, biological systems, the Web,
and any discipline that focuses on modeling interactions between
entities~\cite{easley2012networks,albert2002statistical,newman2003structure}.
The simple model of nodes and edges provides a powerful and flexible
abstraction, and over time, more expressive models have been developed to
incorporate richer structure in data.
In one direction, models now use more information about the nodes and
edges: multilayer networks capture nodes and edges of different
types~\cite{mucha2010community,kivela2014multilayer},
meta-paths formalize heterogeneous relational structure~\cite{sun2011pathsim,Dong2017metapath},
and graph convolutional networks use node features for prediction tasks~\cite{kipf2017semi}.
In another direction, \emph{group}, \emph{higher-order}, or \emph{multi-way}
interactions between several nodes --- as opposed to pairwise interactions ---
are paramount to the model.
In this space, interaction data is modeled with
hypergraphs~\cite{Zhang-2018-discussion,zhou2007learning,benson2019three},
tensors~\cite{acar2009link,papalexakis2014spotting,arrigo2019framework},
affiliation networks~\cite{lattanzi2009affiliation},
simplicial complexes~\cite{osting2017spectral,benson2018simplicial,salnikov2018simplicial,porter2019nonlinearity}, and
motif representations~\cite{benson2016higher,rossi2018higher}.
Designing methods that effectively analyze the richer structure encoded by these
expressive models is an ongoing challenge in graph mining and machine learning.

In this work, we focus on the fundamental problem of clustering,
where the general idea is to group nodes based on some similarity score.
While graph clustering methods have a long history~\cite{schaeffer2007graph,fortunato2010community,Leskovec2010,Moore-2017-community},
existing approaches for rich graph data do not
naturally handle networks with categorical edge labels.
In these settings, a categorical edge label encodes a type
of discrete similarity score --- two nodes connected by an edge
with category label $c$ are similar with respect to $c$.
This structure arises in a variety of settings:
brain regions are connected by different types of connectivity patterns~\cite{crossley2013cognitive};
edges in coauthorship networks are categorized by publication venues, and
copurchasing data can contain information about the type of shopping trip.
In the examples of coauthorship and copurchasing, the interactions are also higher-order --- publications
can involve multiple authors and purchases can be made up of several items.
Thus, we would like a scalable approach to clustering nodes using
a similarity score based on categorical edge labels that work
well for higher-order interactions.

Here, we solve this problem with a novel clustering framework for edge-labeled graphs.
Given a network with $k$ edge labels (categories or colors), we create $k$ clusters of nodes, 
each corresponding to one of the labels.
As an objective function for cluster quality, we seek to simultaneously minimize two quantities:
(i) the number of edges that cross cluster boundaries, and
(ii) the number of intra-cluster ``mistakes'', where an edge of one category is placed inside 
the cluster corresponding to another category.
This approach results in a clustering of nodes that respects 
both the coloring induced by the edge labels and the topology of the original network.
We develop this computational framework in a way that seamlessly generalizes to the case of hypergraphs
to model higher-order interactions, where hyperedges have categorical labels.

The style of our objective function is related to correlation clustering in signed networks~\cite{BansalBlumChawla2004},
as well as its generalization for discrete labels (colors), chromatic correlation clustering~\cite{Bonchi2012ccc,Bonchi2015ccc},
which are based on similar notions of mistake minimization.
However, a key difference is that our objective function does not penalize placing nodes not connected by an edge in the same cluster.
This modeling difference provides serious advantages in terms of tractability, scalability, and the ability to generalize to higher-order interactions.

We first study the case of edge-labeled (edge-colored) graphs with only two categories.
We develop an algorithm that optimizes our \obj{} objective function in polynomial
time by reducing the problem to a minimum $s$-$t$ graph cut problem on a related network.
We then generalize this construction to facilitate quickly finding the optimal solution exactly for hypergraphs.
This is remarkable on two fronts.
First, typical clustering objectives such as minimum bisection, ratio cut, normalized cut, and modularity are NP-hard to optimize
even in the case of two clusters~\cite{wagner1993between,brandes2007modularity}.
And in correlation clustering, having two edge types is also NP-hard~\cite{BansalBlumChawla2004}.
In contrast, our setup admits a simple algorithm based on minimum $s$-$t$ cuts.
Second, our approach seamlessly generalizes to hypergraphs.
Importantly, we do not approximate hyperedge cuts with weighted graph cuts,
which is a standard heuristic approach in hypergraph clustering~\cite{Agarwal2006,zhou2007learning,panli2017inhomogeneous}.
Instead, our objective exactly models the number of hyperedges that cross cluster boundaries and the number
of intra-cluster ``mistake'' hyperedges.

With more than two categories, we show that minimizing our objective is NP-hard, and
we proceed to construct several approximation algorithms.
The first set of algorithms are based on practical linear programming relaxations,
achieving an approximation ratio of $\min \left\{2 - \frac{1}{k}, 2 - \frac{1}{r+1}  \right\}$,
where $k$ is the number of categories and $r$ is the maximum hyperedge size ($r = 2$ for the graph case).
The second approach uses a reduction to multiway cut, where practical algorithms have a
$\frac{r + 1}{2}$ approximation ratio and algorithms of theoretical interest have a $2(1 - \frac{1}{k})$ approximation ratio.

We test our methods on synthetic benchmarks as well as a variety of real-world
datasets coming from neuroscience, biomedicine, and social and information
networks; our methods work far better than baseline approaches at minimizing our
objective function. Surprisingly, our linear programming relaxation often
produces a rounded solution that matches the lower bound, i.e., it exactly
minimizes our objective function. Furthermore, our algorithms are also fast in
practice, often taking under 30 seconds on large hypergraphs.

We examine an application to a variant of the community detection problem
where edge labels indicate that two nodes are in the same cluster and find that
our approach accurately recovers ground truth clusters.  We also show how
our formulation can be used for temporal community detection, in which
one clusters the graph based on topology and temporal consistency.
In this case, we treat binned edge timestamps as categories, and
our approach finds good clusters in terms of topological
metrics \emph{and} temporal aggregation metrics. Finally, we provide a case
study in exploratory data analysis with our methods using cooking data, where a
recipe's ingredients form a hyperedge and its edge label the cuisine type.

\section{Preliminaries and related work}
Let $G = (V,E,C, \ell)$ be an edge-labeled (hyper)graph, where
$V$ is a set of nodes, 
$E$ is a set of (hyper)edges,
$C$ is a set of categories (or colors), and
$\ell\colon E \rightarrow C$ is a function which labels every edge with a category.
Often, we just use $C = \{1,2,\ldots,k\}$, and we can
think of $\ell$ as a coloring of the edges.
We use ``category'', ``color'', and ``label'' interchangeably, as these terms
appear in different types of literature (e.g., ``color'' is common for discrete labeling in graph
theory and combinatorics).
We use $k = \lvert C \rvert$ to denote the number of categories,
$E_c \subseteq E$ for the set of edges having label $c$, and
$r$ for the maximum hyperedge size (i.e., \emph{order}), where the size
of a hyperedge is the number of nodes it contains (in the case of graphs, $r=2$).

\subsection{Categorical edge clustering objective}
Given $G$, we consider the task of assigning a category (color) to each node in such a way that nodes in 
category $c$ tend to participate in edges with label $c$; in this setup, we partition the nodes into $k$ clusters 
with one category per cluster.
We encode the objective function as minimizing the number of ``mistakes'' in a clustering, where a mistake is
an edge that either (i) contains nodes assigned to different clusters or (ii) is placed in a cluster corresponding
to a category which is not the same as its label.
In other words, the objective is to minimize the number of edges that are not completely contained in the cluster corresponding to the edge's label.

Let $\clustering$ be a categorical clustering, or equivalently, a coloring of the nodes, where
$\clustering[i]$ denotes the color of node $i$. 
Let $m_{\clustering}\colon E \rightarrow \{0,1\}$ be the \emph{category-mistake} function, defined for an edge $e \in E$ by
\begin{equation}
m_{\clustering} (e) = \begin{cases} 1 & \text{ if $\clustering[i] \neq \ell(e) $ for any node $i \in e$,} \\
0 & \text{ otherwise.}
\end{cases}
\end{equation}
Then, the \emph{Categorical Edge Label Clustering} objective score for the clustering $\clustering$ is
simply the number of mistakes:
\begin{equation} \label{eq:chromec}
\textstyle \cc(\clustering) = \sum_{e \in E} m_{\clustering}(e).
\end{equation}
This form applies equally to hypergraphs; a mistake is a hyperedge with
a node placed in a category different from the edge's label.

Our objective can easily be modified for weighted (hyper)graphs.
If a hyperedge $e$ has weight $w_e$, then the category mistake function simply becomes $m_{\clustering}(e)=w_e$ if $\clustering[i]\neq\ell(e)$ for any node $i$ in $e$ and is $0$ otherwise.
Our results easily generalize to this setting, but we present results in the unweighted case for ease of notation.

\subsection{Relation to Correlation Clustering}
Our objective function is related to chromatic correlation clustering~\cite{Bonchi2015ccc}, in which one clusters an edge-colored
graph into any number of clusters, and a penalty is incurred for any one of three types of \emph{mistakes}:
(i) an edge of color $c$ is placed in a cluster of a different color; 
(ii) an edge of any color has nodes of two different colors; or 
(iii) a pair of nodes \emph{not} connected by an edge is placed inside a cluster.
This objective is a strict generalization of the classical correlation clustering objective~\cite{BansalBlumChawla2004}.

Our \obj{} objective is similar, except we remove the penalty for placing non-adjacent nodes in the same cluster (mistakes of type (iii)). 
The chromatic correlation clustering objective treats the absence of an edge between nodes $i$ and $j$ as a strong indication that these nodes should not share the same label.
We instead interpret a non-edge simply as missing information: the absence of an edge may be an indication that $i$ and $j$ do not belong together, but it may also be the case that they have a relationship that simply has not been measured.
This is a natural assumption with large, sparse real-world graphs, where we rarely have information on all pairs of entities. 
Another key difference between chromatic correlation clustering and our objective is that in the former,
one may form several clusters for the same color.
For our objective, merging two separate clusters for the same color can only improve the objective.

Our formulation also leads to several differences in computational tractability.
Chromatic correlation clustering is NP-hard in general, and there are several approximation algorithms~\cite{Bonchi2012ccc,Bonchi2015ccc,Anava:2015:ITP:2736277.2741629}.
The tightest of these is a $4$-approximation, though the algorithm is mostly of theoretical interest, as it involves solving an incredibly large linear program.
Moreover, the higher-order generalization of simple correlation clustering (without colors) to hypergraphs is more complicated to solve and approximate than standard correlation clustering~\cite{gleich2018ccgen,fukunga2018highcc,Li2017motifcc,li2018motifCC}.
We will show that our \obj{} objective can be solved in polynomial time for graphs and hypergraphs with two categories.
The problem becomes NP-hard for more than two categories, but we are able to obtain practical 2-approximation algorithms for both graphs and hypergraphs.
Our approaches are based on linear programming relaxations that are small enough to be solved quickly in practice.

\subsection{Additional related work}
There are several methods for clustering general data points that
have categorical features~\cite{ganti1999cactus,gibson2000clustering,boriah2008similarity},
but these methods are not designed for clustering graph data.
There are also methods for clustering in graphs with attributes~\cite{Xu-2012-model,Akoglu-2012-pics,Zhou-2009-clustering,Bothorel-2015-attributed};
these focus on vertex features and do not connect categorical features to cluster indicators.
Finally, there are several clustering approaches for multilayer networks
modeling edge types~\cite{mucha2010community,DeDomenico-2015-reducibility,Lancichinetti-2012-consensus},
but the edge types are not meant to be indicative of a cluster type.

\section{The case of two categories}\label{sec:twocats}
In this section we design algorithms to solve the \obj{}
problem when there are only two categories. In this case, both the graph and hypergraph
problem can be reduced to a minimum $s$-$t$ cut problem, which can be efficiently solved.

\subsection{An algorithm for graphs}\label{sec:graph2}
To solve the two-category problem on graphs, we first convert it to an instance of a weighted
minimum $s$-$t$ cut problem on a graph with no edge labels. 
Recall that $E_c$ is the set of edges with category label $c$.
Given the edge-labeled graph $G = (V, E, C, \ell)$, we construct a new graph $G' = (V', E')$ as follows:
\begin{itemize}
	\item Introduce a terminal node $v_c$ for each of the two labels $c \in L$, 
	so that $V' = V\cup V_t$ where $V_t = \{v_c \given c \in L\}$.
	\item For each label $c$ and each $(i,j) \in E_c$, introduce edges $(i,j)$, $(v_c, i)$ and $(v_c,j)$, 
	all of which have weight $\frac{1}{2}$.
\end{itemize}
Since there are only two categories $c_1$ and $c_2$, let $s = v_{c_1}$ be treated as a source node and 
$t = v_{c_2}$ be treated as a sink node. The minimum $s$-$t$ cut problem in $G'$ is defined by
\begin{equation}
\label{minstcut}
\minimize_{S \subseteq V} \,\, \cut(S\cup{s}),
\end{equation}
where $\cut(T)$ is the weight of edges crossing from nodes in
$T \subset V'$ to its complement set $\bar{T} = V' \backslash T$. 
This classical problem that can be efficiently solved in polynomial time,
and we have an equivalence with the original two-category edge clustering objective.
\begin{proposition}
For any $S \subseteq V$, the value of $\cut(S\cup s)$ in $G'$ is equal to the value of $\cc(\{S, \bar{S} \})$, 
where $S$ and $\bar{S}$ are the clusters for categories $c_1$ and $c_2$.
\end{proposition}
\begin{proof}
Let edge $e = (i,j)$ be a ``mistake'' in the clustering ($m_{\clustering}(e) = 1$) and without loss of generality have color $c_1$.
If $i$ and $j$ are assigned to $c_2$, then the half-weight edges
$(i,v_{c_1})$ and $(j,v_{c_1})$ are cut.
Otherwise, exactly one of $i$ and $j$ is assigned to $c_2$.
Without loss of generality, let it be $i$.
Then $(i, v_{c_1})$ and $(i,j)$ are cut.
\end{proof}
Thus, a minimizer for the $s$-$t$ cut in $G'$ directly gives us a minimizer for
our \obj{} objective.
We next provide a similar reduction for
the case of hypergraphs.

\subsection{An algorithm for hypergraphs}
\begin{figure}[t]
    \input{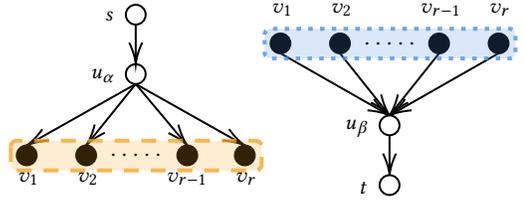}
    \caption{Subgraphs used for the $s$-$t$ cut reduction of two-color \obj{} in hypergraphs.
    Here, $\alpha$ and $\beta$ are hyperedges in the original hypergraph 
    with colors $c_1$ (orange, left) and $c_2$ (blue, right).}
    \label{fig:subgraphs}
\end{figure}
We now develop a method to exactly solve our objective in the two-color case with arbitrary order-$r$ hypergraphs,
and we again proceed by reducing to an $s$-$t$ cut problem.
Our approach is to construct a subgraph for every hyperedge and
paste these subgraphs together to create a new graph $G' = (V',E')$,
where minimum $s$-$t$ cuts produce partitions that minimize the \obj{} objective.
A similar construction has been used for a $\mathcal{P}^{r}$ Potts model in computer vision~\cite{Kohli2007},
and our reduction is the first direct application of this approach to network analysis.

We start by adding terminal nodes $s = v_{c_1}$ and $t = v_{c_2}$ (corresponding to categories $c_1$ and $c_2$)
as well as all nodes in $V$ to $V'$.
For each hyperedge $e = (v_1, \ldots, v_r)$ of $G$, we add a node $u_e$ to $V'$
and add the following \emph{directed} edges to $E'$ (see also \Cref{fig:subgraphs}):
\begin{itemize}
    \item If $e$ has label $c_1$, add $(s, u_e)$, $(u_e,v_1), \ldots, (u_e,v_r)$ to $E'$.
    \item If $e$ has label $c_2$, add $(u_e,t)$, $(v_1,u_e), \ldots, (v_r,u_e)$ to $E'$.
\end{itemize}
Again, the minimum $s$-$t$ cut on $G'$ produces a partition that also minimizes the categorical edge 
clustering objective, as shown below.
\begin{theorem}
Let $\mathcal{S^*}$ be the solution to the minimum cut problem.
Then the label assignment $\clustering$ defined by $\clustering[i] = c_1$ if $i \in S^*$
and $\clustering[i] = c_2$ if $i \in \bar{S}^*$ minimizes the \obj{} objective.
\end{theorem}
\begin{proof}
Consider a hyperedge $e = (v_1, \ldots, v_r)$ with label $c_2$. 
We show that $m_{\clustering}(e)$ precisely corresponds to an $s$-$t$ cut on the subgraph of $G'$
induced by $e$ (\cref{fig:subgraphs}, right).
If $\clustering[v_1] = \ldots = \clustering[v_r] = c_2$, then $v_1, \ldots, v_r \in \bar{S}^*$
and the cost of the minimum $s$-$t$-cut is 0 (via placing $s$ by itself).
Now suppose at least one of $\clustering[v_1], \ldots, \clustering[v_r]$ equals $c_1$. 
Without loss of generality, say that $\clustering[v_1]=c_1$, so $v_1 \in S^*$.
If $u_e \in S^*$, we cut $(u_e,t)$ and none of the edges $(v_i, u_e)$ contribute to the cut.
If $u_e \in \bar{S}^*$, we cut $(v_1,u_e)$; and it cannot be the
case that $(v_i, u_e)$ is cut for $i \neq 1$ (otherwise, we could have reduced the cost of the minimum cut
by placing $u_e \in S^*$).

To summarize, if edge $e$ with label $c_2$ induces a mistake in the clustering,
then the cut contribution is $1$; otherwise, it is $0$.
A symmetric argument holds if $e$ has label $c_1$, using the graph in \cref{fig:subgraphs} (left).
By additivity, minimizing the $s$-$t$ cut in $G'$ minimizes the number of mistakes in the \obj{} objective.
\end{proof}
This procedure also works for the special case of graphs.
However, $G'$ has more nodes and directed edges in the more general reduction,
which can increase running time in practice.

\xhdr{Computational considerations}
Both algorithms solve a single minimum cut
problem on a graph with $O(T)$ vertices and $O(T)$ edges, where
$T = \sum_{e \in E} \lvert e \rvert$ is the sum of hyperedge degrees (this is
bounded above by $r\lvert E \rvert$, where $r$ is the order of the hypergraph).
In theory, this can be solved in $O(T^2)$ time in the worst case~\cite{orlin2013max}.
However, practical performance is often much different than this worst-case running time.
That being said, we do find the maximum flow formulations to often be slower than the
linear programming relaxations we develop in \cref{sec:morecats}.
We emphasize that being able to solve the \obj{} objective in polynomial time for two colors is itself interesting,
and that the algorithms we use for experiments in \cref{sec:experiments} are able to scale to large hypergraphs.

\xhdr{Considerations for unlabeled edges}
Our formulation assumed that all of the (hyper)edges carry a unique label.
However, in some datasets, there may be edges with no label or both labels.
In these cases, the edge's existence still signals that its constituent nodes
should be colored the same --- just not with a particular color.
A natural augmentation to our objective is then to penalize
this edge only when it is not entirely contained in \emph{some} cluster.
Our reductions above handle this case by simply connecting the corresponding
nodes in $V'$ to both terminals instead of just one.

\section{More than two categories}\label{sec:morecats}
We now move to the general formulation of \obj{}
when there can be more than two categories or labels.
We first show that optimizing the objective in this setting is NP-hard.
After, we develop approximation algorithms based on linear programming relaxations and
multiway cut problems with theoretical guarantees on solution quality.
Many of these algorithms are practical, and we use them in numerical experiments in \cref{sec:experiments}.

\subsection{NP-hardness of \obj{}}
We now prove that the \obj{} objective is NP-hard for the case of three categories. 
Our proof follows the structure of the NP-hardness reduction for 3-terminal multiway cut~\cite{Dahlhaus94thecomplexity},
and the reduction is from the NP-hard maximum cut (maxcut) problem.
Written as a decision problem, this problem seeks to answer if there exists
a partition of the nodes of a graph into two sets such that the number of edges cut
by the partition is at least $K$.

\begin{figure}[t]
	\centering
	\scalebox{1.5}{\tikzset{every picture/.style={line width=0.75pt}} 

\begin{tikzpicture}[x=0.75pt,y=0.75pt,yscale=-1,xscale=1]

\draw [line width=1.25]   (125,225) -- (145,205) ;

\draw  [color={rgb, 255:red, 0; green, 0; blue, 0 }  ,draw opacity=1 ][fill={rgb, 255:red, 255; green, 255; blue, 255 }  ,fill opacity=1 ] (120,225) .. controls (120,222.24) and (122.24,220) .. (125,220) .. controls (127.76,220) and (130,222.24) .. (130,225) .. controls (130,227.75) and (127.76,229.99) .. (125,229.99) .. controls (122.24,229.99) and (120,227.75) .. (120,225) -- cycle ;
\draw  [color={rgb, 255:red, 0; green, 0; blue, 0 }  ,draw opacity=1 ][fill={rgb, 255:red, 255; green, 255; blue, 255 }  ,fill opacity=1 ] (140,205) .. controls (140,202.24) and (142.24,200) .. (145,200) .. controls (147.76,200) and (150,202.24) .. (150,205) .. controls (150,207.75) and (147.76,209.99) .. (145,209.99) .. controls (142.24,209.99) and (140,207.75) .. (140,205) -- cycle ;
\draw [color={rgb, 255:red, 74; green, 144; blue, 226 }  ,draw opacity=1, line width=1.25]   (195,240) -- (195,229.99) ;

\draw [color={rgb, 255:red, 126; green, 211; blue, 33 }  ,draw opacity=1, line width=1.25]   (235,220) -- (235,209.99) ;

\draw [color={rgb, 255:red, 126; green, 211; blue, 33 }  ,draw opacity=1, line width=1.25]   (200,245) -- (210,245) ;

\draw [color={rgb, 255:red, 208; green, 2; blue, 27 }  ,draw opacity=1, line width=1.25]   (230,225) -- (200,225) ;

\draw [color={rgb, 255:red, 208; green, 2; blue, 27 }  ,draw opacity=1, line width=1.25]   (215,240) -- (215,209.99) ;

\draw [color={rgb, 255:red, 74; green, 144; blue, 226 }  ,draw opacity=1, line width=1.25]   (230,205) -- (220,205) ;

\draw    [line width=1.25] (156,225) -- (179,225) ;
\draw [shift={(181,225)}, rotate = 180] [fill={rgb, 255:red, 0; green, 0; blue, 0 }  ] [draw opacity=0] (10.72,-5.15) -- (0,0) -- (10.72,5.15) -- (7.12,0) -- cycle    ;

\draw  [color={rgb, 255:red, 0; green, 0; blue, 0 }  ,draw opacity=1 ][fill={rgb, 255:red, 255; green, 255; blue, 255 }  ,fill opacity=1 ] (210,205) .. controls (210,202.24) and (212.24,200) .. (215,200) .. controls (217.76,200) and (220,202.24) .. (220,205) .. controls (220,207.75) and (217.76,209.99) .. (215,209.99) .. controls (212.24,209.99) and (210,207.75) .. (210,205) -- cycle ;
\draw  [color={rgb, 255:red, 0; green, 0; blue, 0 }  ,draw opacity=1 ][fill={rgb, 255:red, 255; green, 255; blue, 255 }  ,fill opacity=1 ] (190,225) .. controls (190,222.24) and (192.24,220) .. (195,220) .. controls (197.76,220) and (200,222.24) .. (200,225) .. controls (200,227.75) and (197.76,229.99) .. (195,229.99) .. controls (192.24,229.99) and (190,227.75) .. (190,225) -- cycle ;
\draw  [color={rgb, 255:red, 0; green, 0; blue, 0 }  ,draw opacity=1 ][fill={rgb, 255:red, 255; green, 255; blue, 255 }  ,fill opacity=1 ] (210,245) .. controls (210,242.24) and (212.24,240) .. (215,240) .. controls (217.76,240) and (220,242.24) .. (220,245) .. controls (220,247.75) and (217.76,249.99) .. (215,249.99) .. controls (212.24,249.99) and (210,247.75) .. (210,245) -- cycle ;
\draw  [color={rgb, 255:red, 0; green, 0; blue, 0 }  ,draw opacity=1 ][fill={rgb, 255:red, 255; green, 255; blue, 255 }  ,fill opacity=1 ] (230,225) .. controls (230,222.24) and (232.24,220) .. (235,220) .. controls (237.76,220) and (240,222.24) .. (240,225) .. controls (240,227.75) and (237.76,229.99) .. (235,229.99) .. controls (232.24,229.99) and (230,227.75) .. (230,225) -- cycle ;
\draw  [color={rgb, 255:red, 0; green, 0; blue, 0 }  ,draw opacity=1 ][fill={rgb, 255:red, 255; green, 255; blue, 255 }  ,fill opacity=1 ] (230,205) .. controls (230,202.24) and (232.24,200) .. (235,200) .. controls (237.76,200) and (240,202.24) .. (240,205) .. controls (240,207.75) and (237.76,209.99) .. (235,209.99) .. controls (232.24,209.99) and (230,207.75) .. (230,205) -- cycle ;
\draw  [color={rgb, 255:red, 0; green, 0; blue, 0 }  ,draw opacity=1 ][fill={rgb, 255:red, 255; green, 255; blue, 255 }  ,fill opacity=1 ] (190,245) .. controls (190,242.24) and (192.24,240) .. (195,240) .. controls (197.76,240) and (200,242.24) .. (200,245) .. controls (200,247.75) and (197.76,249.99) .. (195,249.99) .. controls (192.24,249.99) and (190,247.75) .. (190,245) -- cycle ;

\draw (127,255.83) node [scale=0.7]  {$\mathrm{Edge}$};
\draw (125,225) node [scale=0.7]  {$u$};
\draw (145,205) node [scale=0.7]  {$v$};
\draw (216,255.83) node [scale=0.7]  {$3\mathrm{-color\ gadget}$};
\draw (215,205) node [scale=0.7]  {$v$};
\draw (195,225) node [scale=0.7]  {$u$};

\end{tikzpicture}}
	\caption{Gadget used for reducing maxcut to 3-color \obj{}.
          Each gadget has new auxiliary nodes, but $u$ and $v$ may be a part of many 3-color gadgets.} 
	\label{fig:3colorgadget}
\end{figure}
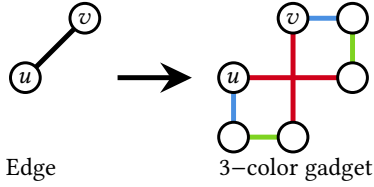

Consider an unweighted instance of maxcut on a graph $G = (V,E)$. 
To convert this into an instance of 3-color \obj{}, we replace each edge 
$(u,v) \in E$ with the 3-color gadget in \cref{fig:3colorgadget}.
We will use the following lemma in our reduction.
\begin{lemma} \label{lem:3color}
	In any node coloring of the 3-color gadget (\cref{fig:3colorgadget}), the minimum number
	of edges whose color does not match both of its nodes (i.e., number of mistakes in categorical edge clustering) is three.
	This only occurs when one of $\{u,v\}$ is red and the other is blue.
\end{lemma}
\begin{proof}
	If $v$ is blue and $u$ is red, then we can achieve the minimum three mistakes 
	by clustering each node in the gadget with its horizontal neighbor in \cref{fig:3colorgadget}
        or alternatively by placing each node with its vertical neighbor.
	If $u$ and $v$ are constrained to be in the same cluster, then the optimal solution is to place all nodes in the gadget together, which makes 4 mistakes.
        It is not hard to check that all other color assignments yield a penalty of 4 or more.
\end{proof}

Now let $G'$ be the instance of 3-color \obj{} obtained by replacing each edge $(u,v) \in E$ 
with a 3-color gadget.
\begin{theorem}
	There exists a partition of the nodes in $G$ into two sets with $K$ or more cut edges 
	if and only if there is a 3-coloring of the nodes in $G'$ that makes $4 \lvert E \rvert - K$ or fewer mistakes.
\end{theorem}
\begin{proof}
  Consider first a cut in $G = (V,E)$ of size $K' \geq K$.
  Let $S_r$ and $S_b$ denote the two clusters in the corresponding bipartition of $G$, mapping to red and blue clusters.
  Consider each $(u,v) \in E$ in turn along with its 3-color gadget.
  If $(u,v) \in E$ is cut, cluster all nodes in its gadget with their vertical neighbor if $u \in S_b$ and $v \in S_r$, and cluster them with their horizontal neighbor if $u \in S_r$ and $v \in S_b$.
  Either way, this makes exactly 3 mistakes. If $(u,v)$ is \emph{not} cut, then label all nodes in the gadget red if $u,v \in S_r$, or blue if $u,v \in S_b$, which makes exactly 4 mistakes.
  The total number of mistakes in $G'$ is then $3K' + 4(\lvert E \rvert - K') = 4\lvert E \rvert - K' \leq 4 \lvert E \rvert - K$.

  Now start with $G'$ and consider a node coloring that makes $B' \leq B = 4\lvert E \rvert - K$ mistakes. 
  There are $\lvert E \rvert$ total 3-color gadgets in $G'$. 
  We claim that there must be at least $K$ of these gadgets at which only three mistakes are made. 
  If this were not the case, then assume exactly $H < K$ gadgets where 3 mistakes are made. 
  By \cref{lem:3color}, there are $\lvert E \rvert - H$ gadgets where at least 4 mistakes are made,
  so the total number of mistakes is
  $B' \geq 3H + 4(\lvert E \rvert - H) = 4\lvert E \rvert - H > 4\lvert E \rvert - K$,
  contradicting our initial assumption. Thus, by \cref{lem:3color}, there are at least $K$ edges $(u,v) \in E$ 
  where one of $\{u,v\}$ is red and the other is blue,
  and the maximum cut in $G$ is at least $K$.
\end{proof}

Consequently, if we can minimize \obj{} in polynomial time, 
we can solve the maximum cut decision problem in polynomial time,
and \obj{} is thus NP-hard.
As a natural next step, we turn to approximation algorithms.

\subsection{Algorithms based on LP relaxations}
\label{sec:LP}

We now develop approximation algorithms by relaxing an integer linear programming (ILP) formulation of our problem.
We design the algorithms for hypergraphs, with graphs as a special case.
Suppose we have an edge-labeled hypergraph $G = (V,E,C,\ell)$ with $C = \{1,\ldots,k\}$,
where $E_c = \{e \in E \given \ell[e] = c\}$.
The \obj{} objective can be written as the following ILP:
\begin{equation}
\label{eq:ccilp}
\begin{array}{lll}
\min & \sum_{c \in C} \sum_{e \in E_c} x_e & \\[1mm]
\text{s.t.} & \text{for all $v \in V$:} & \sum_{c = 1}^k x_v^c = k -1\\
& \text{for all $c \in C$, $e \in E_c$:} & x_v^c \leq x_e \text{ for all $v \in e$} \\
& x_v^c, x_e \in \{0,1\} &\text{for all $c \in C$, $v \in V$, $e \in E$.}
\end{array}
\end{equation}
In this ILP, $x_v^c = 1$ if node $v$ is \emph{not} assigned to category $c$, and is zero otherwise. 
The first constraint in~\eqref{eq:ccilp} ensures that $x_v^c = 0$ for exactly one category. 
The second constraint says that in any minimizer, $x_e = 0$ if and only if all nodes in $e$ are colored the same as $e$;
otherwise, $x_e = 1$.
If we relax the binary constraints in~\eqref{eq:ccilp}:
\begin{align*}
0 \leq x_v^c \leq 1,\quad 0 \leq x_e \leq 1,
\end{align*} 
then the ILP is just a linear program (LP) that can be solved in polynomial time. 

When $k = 2$, the constraint matrix of the LP relaxation is totally unimodular as it corresponds to the incidence matrix of a balanced signed graph~\cite{zaslavsky}. 
Thus, all basic feasible solutions for the LP satisfy the binary constraints of the original ILP~\eqref{eq:ccilp},
which is another proof that the two-category problem can be solved in polynomial time.

With more than two categories, the LP solution can be fractional,
and we cannot directly determine a node assignment from the LP solution. 
Nevertheless, solving the LP provides a lower bound on the optimal solution, 
and we show how to round the result to produce a clustering within a bounded factor of the lower bound.
\Cref{alg:2round} contains our rounding scheme, and the following theorem shows that it provides
a clustering within a factor of 2 from optimal.

\begin{algorithm}[tb]
\DontPrintSemicolon
\caption{A simple 2-approximation for \obj{} based on an LP relaxation.
  \Cref{alg:rand} details a more sophisticated rounding scheme.}
\label{alg:2round}
{\bfseries Input:} Labeled hypergraph $G = (V,E,C,\ell)$.\;
{\bfseries Output:} Label $\clustering[i]$ for each node $i \in V$. \;
Solve the LP-relaxation of ILP~\eqref{eq:ccilp}.\;
\For{$c \in C$}{
$S_c \leftarrow \{v \in V \given x_v^c < 1/2\}$.\;
\lFor{$i \in S_c$}{assign $\clustering[i] \leftarrow c$.}
}
Assign unlabeled nodes to an arbitrary $c \in C$.\;

\end{algorithm}

\begin{theorem}
	\Cref{alg:2round} returns at worst a 2-approximation to the \obj{} objective.
\end{theorem}
\begin{proof}
	First, for any $v \in V$, $x_v^c < 1/2$ for at most one category $c \in C$
	in the solution. If this were not the case, there would exist two colors $a$ and $b$ such that 
	$x_v^{a} < 1/2$ and $x_v^b < 1/2$ and
	\[ \textstyle
	 \sum_{c = 1}^k x_v^c = x_v^a + x_v^b + \sum_{c' \in C \backslash\{a,b\}} x_v^{c'} < 1 + k - 2 = k - 1,
	 \]
	 which violates the first constraint of the LP relaxation.
         Therefore, each node will be assigned to at most one category. 
	Consider any $e \in E_c$ for which all nodes are not assigned to $c$. 
	This means that there exists at least one node $v \in e$ such that $x_v^c \geq 1/2$. 
	Thus, the Algorithm incurs a penalty of one for this edge, 
	but the LP relaxation pays a penalty of $x_e \geq x_v^c \geq 1/2$.
	Therefore, every edge mistake will be accounted for within a factor of 2.
\end{proof}

We can get better approximations in expectation with a more sophisticated randomized rounding algorithm (\cref{alg:rand}).
In this approach, we form sets $S_c^t$ based on a threshold parameter $t$ so that each node may be included in more than one set.
To produce a valid clustering, we first generate a random permutation of colors to indicate an (arbitrary) priority of one color over another. 
For any $v \in V$ contained in more than one set $S_c^t$, we assign $v$ to the cluster with highest priority.
By carefully setting the parameter $t$, this approach has better guarantees than \cref{alg:2round}.

\begin{algorithm}[tb]
\DontPrintSemicolon
\caption{LP relaxation for \obj{} with a randomized rounding scheme.
\Cref{thm:randround} gives approximation guarantees based on $t$.}
\label{alg:rand}
{\bfseries Input:} Labeled hypergraph $G = (V,E,C = \{1,2,\ldots,k\}, \ell)$; 
			rounding parameter $t \in \left[1/2, 2/3\right]$. \;
{\bfseries Output:} Label $\clustering[i]$ for each node $i \in V$. \;
Solve the LP-relaxation of ILP~\eqref{eq:ccilp}. \;
$\pi \leftarrow$ uniform random permutation of $\{1,2, \hdots, k\}$.\;
\For{$c = \pi_1, \ldots, \pi_k$}{
	$S_c \leftarrow \{v \in V \given x_v^c < t\}$.\;
	\lFor{$i \in S_{c}$}{$\clustering[i] \leftarrow \pi(c)$.}
}
Assign unlabeled nodes to an arbitrary $c \in C$.\;
\end{algorithm}

\begin{theorem}\label{thm:randround}
	If $t = k / (2k-1)$,
	\cref{alg:rand} returns an at worst $(2 - 1/k)$-approximation for \obj{} in expectation.
	And if $t = (r+1)/(2r+1)$, 
	\cref{alg:rand} returns an at worst $(2 - 1/(1+r))$-approximation in expectation.
\end{theorem}
\begin{proof}
For the choices of $t$ listed in the statement of the theorem, 
$t \in [1/2, 2/3]$ as long as $r \geq 2$ and $k \geq 2$, which is always true.
We say that color $c$ \emph{wants} node $v$ if $v \in S_c$,
but this does not automatically mean that $v$ will be colored as $c$.
For any $v \in V$, there exist at most two colors that want $v$. 
If $v$ were wanted by more than two colors, this would mean $v \in S_a \cap S_b \cap S_c$ for three distinct colors $a,b,c$.
This leads to a violation of the first constraint in~\eqref{eq:ccilp}:
\begin{align*}
	x_v^a + x_v^b  + x_v^c + \sum_{i: i \notin \{a,b,c\}} x_v^i < 3t + (k-3) \leq 2 + (k-3) = (k-1).
\end{align*}

Consider an arbitrary $t \in (1/2, 2/3)$.
We can bound the expected number of mistakes made by \cref{alg:rand} and pay for them individually in terms of the LP lower bound.
To do this, we consider a single hyperedge $e \in E_c$ with color $c$ and bound the probability of making a mistake and the LP cost of this hyperedge.

\emph{Case 1: $x_e \geq t$.}
In this case, we are guaranteed to make a mistake at edge $e$, since $x_e \geq t$ implies there is some node $v \in e$ such that $x_v^c \geq t$, and so $v \notin S_c$. 
However, because the LP value at this edge is $x_e \geq t$, we pay for our mistake within a factor $1/t$.

\emph{Case 2: $x_e < t$.}
Now, color $c$ wants every node in the hyperedge $e \in E_c$. 
If no other colors want any node $v \in e$, then \cref{alg:rand} will not make a mistake at $e$, and we have no mistake to account for. 
Assume then that there is some node $v \in e$ and a color $c' \neq c$ such that $c'$ wants $v$.
This implies that $x_{v}^{c'} < t$, from which we have
that $x_{v}^{c} \geq 1 - x_{v}^{c'}  > 1-t$ (to satisfy the first inequality in~\eqref{eq:ccilp}).
Thus,
\begin{equation}
\label{lpcost}
x_e \geq x_{v}^{c'} > 1-t.
\end{equation}
This gives a lower bound of $1 - t$ on the contribution of the LP objective at edge $e$.

In the worst case, each $v \in e$ may be wanted by a \emph{different} $c' \neq c$, and
the number of colors other than $c$ that want some node in $e$ is bounded above by $B_1 = k - 1$ and $B_2 = r$. 
We avoid a mistake at $e$ if and only if $c$ has higher priority than all of the alternative colors, where priority is established by the random permutation $\pi$.
Thus,
\begin{equation}
\label{probme}
\textstyle \textbf{Pr}[\text{mistake at e} \given x_e < t] \leq \frac{B_i}{B_i+1} = \min \left \{  \frac{r}{r+1}, \frac{k-1}{k} \right\}.
\end{equation}

Recall from~\eqref{lpcost} that the LP pays $x_e > 1-t$. Therefore, the expected cost at a hyperedge $e \in E_c$ satisfying $x_e < t$ is at most $\frac{B_i}{(1-t)(B_i+1)}$ in expectation. 
Taking the worst approximation factor from Case 1 and Case 2, 
\cref{alg:rand} will in expectation provide an approximation factor of
$\max \left \{ \frac{1}{t}, \frac{B_i}{(1-t)(B_i+1)} \right \}$.
This will be minimized when the approximation bounds from Cases 1 and 2 are equal,
which occurs when $t = \frac{B_i+1}{2B_i+1}$.
If $B_i = k-1$, then $t = \frac{k-1}{2k - 1}$ and the expected approximation factor is $2 - 1 / k$.
And if $B_i = r$, then $t = \frac{r}{2r+1}$ and the expected approximation factor is $2 - 1 / (r + 1)$.
\end{proof}
For the graph case ($r = 2$), this theorem implies a $\frac53$-approximation for \obj{} with any number of categories.

\xhdr{Computational considerations}
The linear program has $O(\lvert E \rvert)$ variables
and sparse constraints, which written as a matrix inequality would have
$O(T)$ non-zeros, where $T$ is again the sum of hyperedge degrees.
Improving the best theoretical running times for solving linear programs is an active area of research~\cite{lee2015efficient,cohen2019solving},
but practical performance of solving linear programs is often much different than worst-case guarantees.
In \Cref{sec:experiments}, we show that a high-performance
LP solver from Gurobi is extremely efficient in practice, finding solutions in seconds on hypergraphs with
several categories and tens of thousands of hyperedges in tens of seconds.

\subsection{Algorithms based on multiway cut}\label{sec:mwc}
We now provide alternative approximations based on multiway cut,
similar to the reductions from \cref{sec:twocats}.
Again, we develop this technique for general hypergraphs and graphs are a special case.

Suppose we have an edge-labeled hypergraph $G = (V,E,C,\ell)$.
We construct a new graph $G' = (V', E')$ as follows.
First, introduce a terminal node $v_c$ for each category $c \in C$, so that $V' = V \cup \{ v_c \given c \in C \}$.
Second, for each hyperedge $e=\{v_1,\ldots,v_r\} \in E$,
add a clique on nodes $v_1, \ldots, v_r, v_{\ell[e]}$ to $E'$, where each
edge in the clique has weight $1/r$.
(Overlapping cliques are just additive on the weights.)

The multiway cut objective
is the number of cut edges in any partition of the nodes into $k$
clusters such that each cluster contains exactly one of the terminal nodes.
We can associate each cluster with a category,
and any clustering $\clustering$ of nodes in \obj{} for $G$
can be mapped to a candidate partition for multiway cut in $G'$.
Let $\mwc(\clustering)$ denote the value of the multiway cut objective
for the clustering $\clustering$.
The next result relates multiway cut to \obj{}.
\begin{theorem}\label{thm:mwc}
For any clustering $\clustering$,
\[
\cc(\clustering) \leq \mwc(\clustering) \leq \frac{r + 1}{2} \cc(\clustering).
\]
\end{theorem}
\begin{proof}
Let $e = \{v_1, \ldots, v_r\}$ with label $c = \ell[e]$ be a hyperedge in $G$.
We can show that the bounds hold when considering the associated clique in $G'$
and then apply additivity.
First, if $e$ is not a mistake in the \obj{},
then no edges are cut in the clique.
If $e$ is a mistake in the \obj{},
then there are some edges cut in the associated clique.
The smallest possible contribution to the multiway cut objective occurs
when all but one node is assigned to $c$.
Without loss of generality, consider this to be $v_1$, which is
in $r$ cut edges: $(r-1)$ corresponding to the edges from $v_1$ to other nodes in the hyperedge, plus one for the edge from $v_1$ to the terminal $v_c$. 
Each of the $r$ cut edges has weight $1/r$, so the multiway cut contribution is 1.

The largest possible cut occurs when all nodes in $e$ are colored differently from $e$.
In this case, the edges incident to each node in the clique are all cut.
For any one of these nodes, the sum of edge weights incident to that node equals 1 by the same arguments as above.
This cost is incurred for each of the $r$ nodes in the hyperedge plus the terminal node $v_c$, 
for a total weight of $r + 1$. Since each edge is counted twice, the actual penalty is $(r+1)/2$.
\end{proof}

\xhdr{Computational considerations}
Minimizing the multiway cut objective is NP-hard~\cite{Dahlhaus94thecomplexity},
but there are many approximation algorithms.
\Cref{thm:mwc} implies that any $p$-approximation for multiway cut
provides a $p(r + 1)/2$-approximation for \obj{}.
For example, the simple isolating cuts heuristic yields a $\frac{r + 1}{2}(2 - \frac{2}{k})$-approximation,
and more sophisticated algorithms provide a $\frac{r+1}{2}(\frac{3}{2} - \frac{1}{k})$-approximation~\cite{calinescu2000}.
For our experiments, we use the isolating cut approach, which
solves $O(k)$ maximum flow problems on a graph with
$O(r \lvert E \rvert)$ vertices and $O(r^2\lvert E \rvert)$ edges. This can be expensive in practice.
We will find that the LP relaxation performs better in terms of solution
quality and running time.

\xhdr{A node-weighted multiway cut reduction}
We also provide an approximation based on a {\it direct} reduction to a node-weighted multiway cut (NWMC) problem
that is of theoretical interest.
As above, suppose we have an edge-labeled hypergraph $G = (V,E,C,\ell)$.
We construct a new graph $G' = (V', E')$ as follows.
First, introduce a terminal node $v_c$ for each category $c \in C$, so that $V' = V \cup \{ v_c \given c \in C \}$.
Assign infinite weights to all nodes in $V'$.
Next, for each hyperedge $e=\{v_1,\ldots,v_r\} \in E$, add an auxiliary node $v_e$ with weight 1.
Next, append edges $(v_e,v_1),\ldots,(v_e,v_r)$ as well as $(v_c,v_e)$ for $\ell(e) = c$ to $E'$.
It straightforward to check that deleting $v_e$ corresponds to making a mistake at hyperedge $e$.
Thus an optimizer of NWMC on $G'$ is also an optimizer of \obj{} on $G$.

Solving NWMC is also NP-hard~\cite{garg2004multiway},
and there are again well-known approximation algorithms.
The above discussion implies any $p$-approximation to NWMC also
provides a $p$-approximation for \obj{}.
For example, an LP-based algorithm has a $2(1-1/k)$-approximation~\cite{garg2004multiway}.
This approximation is better but the LPs are too large to be practical;
however, the improvement of a direct algorithm suggests room for better theoretical results.

\subsection{Approximation through a linear objective}\label{sec:linear_mv}
The \obj{} objective assigns a penalty of 1 regardless of the {\it proportion} of the nodes in a hyperedge which are clustered away from hyperedge's color.
Although useful, we might consider alternative penalties that value the \emph{extent} to which each hyperedge is satisfied in the final clustering.
One natural penalty for a hyperedge of color $c$ is the number of nodes within that hyperedge that are not clustered into that color.
With such a ``linear'' mistake function, we define the Categorical Node Clustering Objective as
\[
\textstyle CatNodeClus(Y)=\sum_{e\in E}m'_Y(e), \text{ where } m'_Y(e)=\sum_{i\in e}I_{Y[i]\neq\ell{}(e)}.
\]
It turns out that this objective is optimized with a simple majority vote algorithm
that assigns a node to the majority color of all hyperedges that conatin it.
\begin{theorem}
The majority vote algorithm yields an optimizer of the Categorical Node Clustering (linear) objective.
\end{theorem}
\begin{proof}
  Suppose node $u$ is contained in $J_i$ hyperedges of color $i$.
  Without loss of generality, assume $J_1 \geq \ldots \geq J_k$.
  The cost of assigning $u$ to $c$ is $C_c=\sum_{j\neq c}J_j$,
  which is minimized for $c = 1$.
\end{proof}
In Section~\ref{sec:experiments}, we will see that the majority vote solution provides a good approximation to the optimizer of the \obj{} objective.
The reason is that the cost of a hyperedge under the linear objective is at most $r$ while that cost under the \obj{} objective is just 1,
which makes majority vote an $r$-approximation algorithm.
\begin{theorem}
The majority vote algorithm provides an \\$r$-approximation for \obj{}.
\end{theorem}


\begin{table*}[t!]
	\caption{Summary statistics of datasets --- number of nodes $\lvert V \rvert$, number of (hyper)edges $\lvert E \rvert$, 
	maximum hyperedge size $r$, and number of categories $k$ --- along with \obj{}
	performance for the algorithms \emph{LP-round} (LP), \emph{Majority Vote} (MV), \emph{Cat-IsoCut} (IC), 
	\emph{ChromaticBalls} (CB) and \emph{LazyChromaticBalls} (LCB). 
	Performance is listed in terms of the approximation guarantee given by the LP lower bound (lower is better)
	and in terms of the edge satisfaction, which is the fraction of edges that are \emph{not mistakes}
	(higher is better; see \cref{eq:chromec}).
	Our LP method performs the best overall and can even find exactly (or nearly) optimal solutions
	to the NP-hard objective by matching the lower bound.
	We also report the running times for rough comparison, though our implementations are not optimized for efficiency. Due to its simplicity, MV is extremely fast.
	}
	\label{tab:allexp1}
	\centering
	\scalebox{0.95}{\begin{tabular}{l     l l l l     l l l l l     l l l l l    l l l ll}
		\toprule
		&&&&& \multicolumn{5}{c}{\textbf{Approx. Guarantee}} & \multicolumn{5}{c}{\textbf{Edge Satisfaction}} & \multicolumn{5}{c}{\textbf{Runtime (in seconds)}} \\ 
		\cmidrule(lr){6-10} \cmidrule(lr){11-15}  \cmidrule(lr){16-20}
		\emph{Dataset} & $\lvert V \rvert$ & $\lvert E \rvert$ & $r$ & $k$ & LP  & MV   & IC   & CB   & LCB  & LP   & MV   & IC   & CB   & LCB  & LP & MV& IC & CB & LCB \\
		\midrule
Brain & 638 & 21180 & 2 & 2 &1.0 & 1.01 & 1.27 & 1.56 & 1.41 & 0.64 & 0.64 & 0.55 & 0.44 & 0.5 & 1.8 & 0.0 & 1.9 & 0.4 & 0.8\\
MAG-10 & 80198 & 51889 & 25 & 10 &1.0 & 1.18 & 1.37 & 1.44 & 1.35 & 0.62 & 0.55 & 0.48 & 0.45 & 0.49 & 51 & 0.1 & 203 & 333 & 699\\
Cooking & 6714 & 39774 & 65 & 20 &1.0 & 1.21 & 1.21 & 1.23 & 1.24 & 0.2 & 0.03 & 0.03 & 0.01 & 0.01 & 72 & 0.0 & 1223 & 4.6 & 6.7\\
DAWN & 2109 & 87104 & 22 & 10 &1.0 & 1.09 & 1.0 & 1.31 & 1.15 & 0.53 & 0.48 & 0.53 & 0.38 & 0.46 & 13 & 0.0 & 190 & 0.3 & 0.4\\
Walmart-Trips & 88837 & 65898 & 25 & 44 & 1.0 & 1.2 & 1.19 & 1.26 & 1.26 & 0.24 & 0.09 & 0.09 & 0.04 & 0.05 & 7686 & 0.2& 68801 & 493 & 1503\\
		\bottomrule
	\end{tabular}}
\end{table*} 

\section{Experiments}\label{sec:experiments}
We now run four types of numerical experiments to demonstrate our methodology.
First, we show that our algorithms indeed work well
on a broad range of datasets at optimizing our objective function and discover
that our LP relaxation tends be extremely effective in practice, often finding an
optimal solution (i.e., matching the lower bound). After, we show that our
approach is superior to competing baselines in categorical community detection
experiments where edges are colored to signal same-community membership.
Next, we show how to use timestamped edge information as a categorical edge label, and demonstrate that our method
can find clusters that preserve temporal information better than methods that
only look at graph topology, without sacrificing performance on topological
metrics. Finally, we present a case study on a network of cooking ingredients
and recipes to show that our methods can also be used for exploratory data
analysis. Our code and datasets are available at~\url{https://github.com/nveldt/CategoricalEdgeClustering}.

\subsection{Analysis on Real Graphs and Hypergraphs}
\label{sec:exp1}
We first evaluate our methods on several real-world edge-labeled graphs and
hypergraphs in terms of \obj{}. The purpose of these experiments is to show
that our methods can optimize the objective quickly and accurately and
to demonstrate that our methods find global categorical clustering
structure better than natural baseline algorithms. All experiments ran on a
laptop with a 2.2 GHz Intel Core i7 processor and 8 GB of RAM.
We implemented our algorithms in Julia, using Gurobi software to solve the linear programs.

\xhdr{Datasets}
\Cref{tab:allexp1} provides summary statistics of the datasets we use, and
we briefly describe them.
\emph{Brain}~\cite{crossley2013cognitive} is a graph where nodes represent brain regions from an MRI.
There are two edge categories: one for connecting regions with high fMRI correlation
and one for connecting regions with similar activation patterns.
In the Drug Abuse Warning Network (\emph{DAWN})~\cite{DAWN-data},
nodes are drugs, hyperedges are
combinations of drugs taken by a patient prior to an emergency room visit, and
edge categories indicate the patient disposition (e.g., ``sent home'' or ``surgery'').
The \emph{MAG-10} network is a subset of the Microsoft Academic
Graph~\cite{MAG-data} where nodes are authors, hyperedges correspond to a
publication from those authors, and there are 10 edge categories which denote
the computer science conference publication venue (e.g., ``WWW'' or ``KDD'').
If the same set of authors published at more than one conference, we
used the most common venue as the category, discarding cases where there is a tie.
In the \emph{Cooking} dataset~\cite{cooking-data},
nodes are food ingredients, hyperedges are
recipes made from combining multiple ingredients, and categories indicate cuisine
(e.g., ``Southern-US'' or ``Indian'').
Finally, the \emph{Walmart-Trips} dataset is made up of products (nodes), groups of
products purchased in a single shopping trip (hyperedges), and categories are
44 unique ``trip types'' classified by Walmart~\cite{walmart-data}.

\xhdr{Algorithms}
We use two algorithms that we developed in
\cref{sec:morecats}. The first is the simple 2-approximation rounding scheme
outlined in \cref{alg:2round}, which we refer to as \emph{LP-round} (\emph{LP}) (in
practice, this performs as well as the more sophisticated algorithm in
\cref{alg:rand} and has the added benefit of being deterministic).  The second
is \emph{Cat-IsoCut} (\emph{IC}), which runs the standard isolating cut
heuristic~\cite{Dahlhaus94thecomplexity} on an instance of multiway cut derived
from the \obj{} problem, as outlined in \cref{sec:mwc}.

The first baseline we compare against is \emph{Majority Vote} (\emph{MV})
discussed in \cref{sec:linear_mv}:
node $i$ is assigned to category $c$ if $c$ is the most
common edge type in which $i$ participates. The \emph{MV} result is also
the default cluster assignment for \emph{IC}, since in practice
this method leaves some nodes unattached from all terminal nodes.

The other baselines are \emph{Chromatic Balls} (\emph{CB}) and \emph{Lazy
  Chromatic Balls} (\emph{LCB}) --- two algorithms for chromatic
correlation clustering~\cite{Bonchi2015ccc}. These methods repeatedly select
an unclustered edge and greedily grow a cluster around it by adding nodes
that share edges with the same label. Unlike our methods, \emph{CB} and
\emph{LCB} distinguish between category (color) assignment and cluster
assignment: two nodes may be colored the same but placed in different
clusters. To provide a uniform comparison among methods, we merge distinct
clusters of the same category into one larger cluster. These methods are
\emph{not} designed for hypergraph clustering, but we still use them for
comparison by reducing a hypergraph to an edge-labeled graph, where nodes
$i$ and $j$ share an edge in category $c$ if they appear together in more
hyperedges of category $c$ than any other.

\xhdr{Results}
Table~\ref{tab:allexp1} reports how well each algorithm solves the \obj{} objective.
We report the approximation guarantee (the ratio between each algorithm's output and the LP lower bound),
as well as the \emph{edge satisfaction}, which is the fraction of hyperedges that end up inside a cluster with the correct label.
Maximizing edge satisfaction is equivalent to minimizing the number of edge label mistakes but provides an intuitive way to interpret and analyze our results. 
High edge satisfaction scores imply that a dataset is indeed characterized by large groups of objects that tend to interact in a certain way with other members of the same group.
A low satisfaction score indicates that a single label for each node may be insufficient to capture the intricacies of the data.

In all cases, the LP solution is integral or nearly integral, indicating that \emph{LP} does an extremely good job solving the original NP-hard objective, often finding an exactly-optimal solution. 
As a result, it outperforms all other methods on all datasets. Furthermore, on nearly all datasets, we can solve the LP within a few seconds or a few minutes.
\emph{Walmart} is the exception--given the large number of categories, the LP contains nearly 4 million variables, and far more constraints.
Other baseline algorithms can be faster, but they do not perform as well in solving the objective. 

The high edge satisfaction scores indicate that our method does the best job identifying sets of nodes which \emph{as a group} tend to participate in one specific type of interaction. In contrast, the \emph{MV} algorithm identifies nodes that {individually} exhibit a certain behavior, but the method does not necessarily form clusters of nodes that as a group interact in a similar way.
Because our \emph{LP} method outperforms our \emph{IC} approach 
on all datasets in terms of both speed and accuracy, in the remaining experiments we focus only on comparing \emph{LP} against other competing algorithms.

\subsection{Categorical Edge Community Detection}
Next we demonstrate the superiority of \emph{LP} in detecting communities of nodes with the same node labels (i.e., \emph{categorical communities}), based on labeled edges between nodes. We perform experiments on synthetic edge-labeled graphs, as well as two real-world datasets, where we reveal edge labels indicative of the ground truth node labels and see how well we can recover the node labels.

\xhdr{Synthetic Model}
We use the synthetic random graph model of Bonchi et al.\ for chromatic correlation clustering~\cite{Bonchi2015ccc}.
A user specifies the number of nodes $n$, colors $L$, and clusters $K$, as well as edge parameters $p$, $q$, and $w$.
The model first assigns nodes to clusters uniformly at random, and then assigns clusters to colors uniformly at random.
(Due to the random assignment, some clusters and colors may not be sampled. Thus, $K$ and $L$ are upper bounds on the number of distinct clusters and unique colors.)
For nodes $i$ and $j$ in the same cluster, the model connects them with an edge with probability $p$.
With probability $1-w$, the edge is the same color as $i$ and $j$. Otherwise, it is a uniform random color. 
If $i$ and $j$ are in different clusters, an edge is drawn with probability $q$ and given a uniform random color. 
We will also use a generalization of this model to synthetic $r$-uniform hypergraphs.
The difference is that we assign colored hyperedges to $r$-tuples of the $n$ nodes, rather than just pairs, and we assign each cluster to a unique color.

\xhdr{Synthetic Graph Results}
We set up two experiments, where performance is measured by the fraction of nodes
placed in the correct cluster (node label accuracy).
In the first, we form graphs with $n = 1000$, $p = 0.05$, and $q = 0.01$, fixing $L = K = 15$
(which in practice leads to graphs with 15 clusters and typically between 8 and 12 distinct edge and cluster colors).
We then vary the noise parameter $w$ from $0$ to $0.75$ in increments of $0.05$. 
\Cref{fig:noise} reports the median accuracy over 5 trials of each method for each value of $w$.
In the second, we fix $w = 0.2$, and vary the number of clusters $K$ from 5 to 50 in increments of 5 with $L = K$.
\Cref{fig:clusters} reports the median accuracy over 5 trials for each value of $K$.

For our first two experiments, we additionally found that our LP algorithm similarly outperformed other methods in terms of cluster identification scores such as Adjusted Rand Index and F-score, followed in performance by MV. Cluster identification scores for LCB and CB were particularly low (ARI scores always below 0.02), as these methods tended to form far too many clusters.

\begin{figure}[t]
	\centering
	\subfloat[Graphs: Varying Noise $w$ \label{fig:noise}]
	{\includegraphics[width=.475\linewidth]{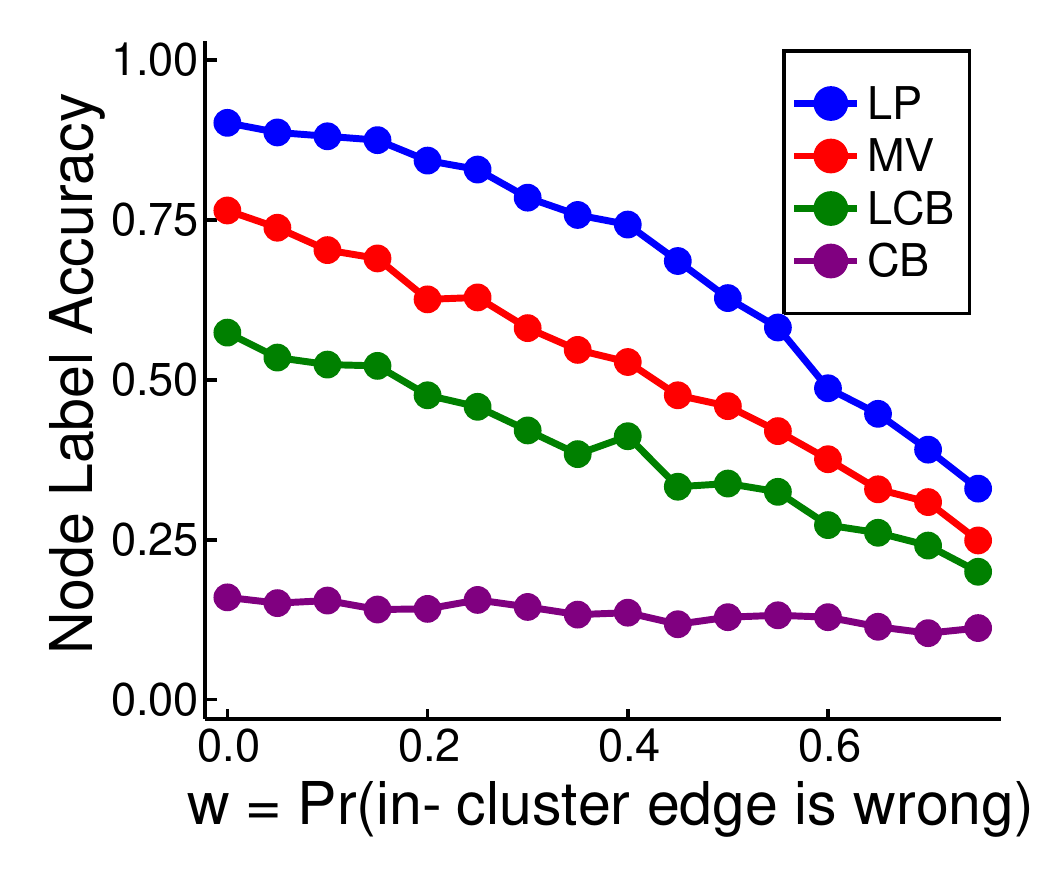}}\hfill
	\subfloat[Graphs: Varying \# Clusters \label{fig:clusters}]
	{\includegraphics[width=.475\linewidth]{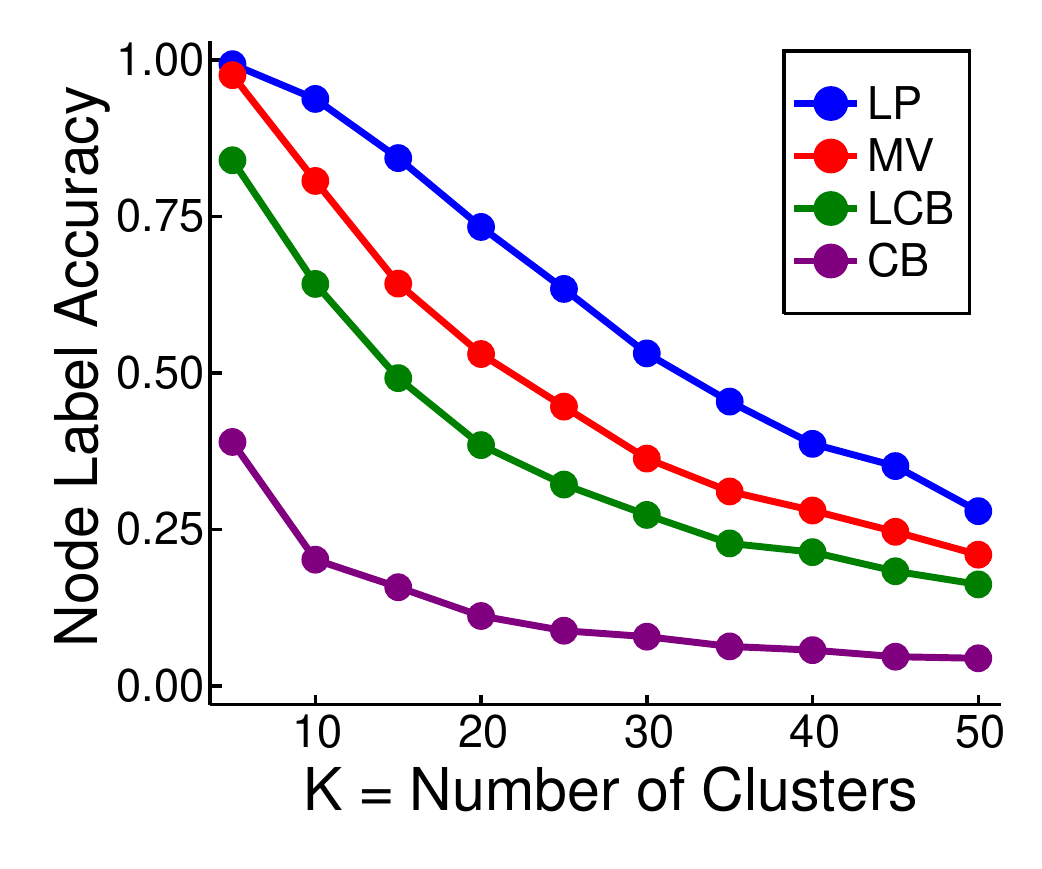}}\hfill
	\subfloat[Hypergraphs: Varying Noise $w$ \label{fig:noise-hyper}]
	{\includegraphics[width=.475\linewidth]{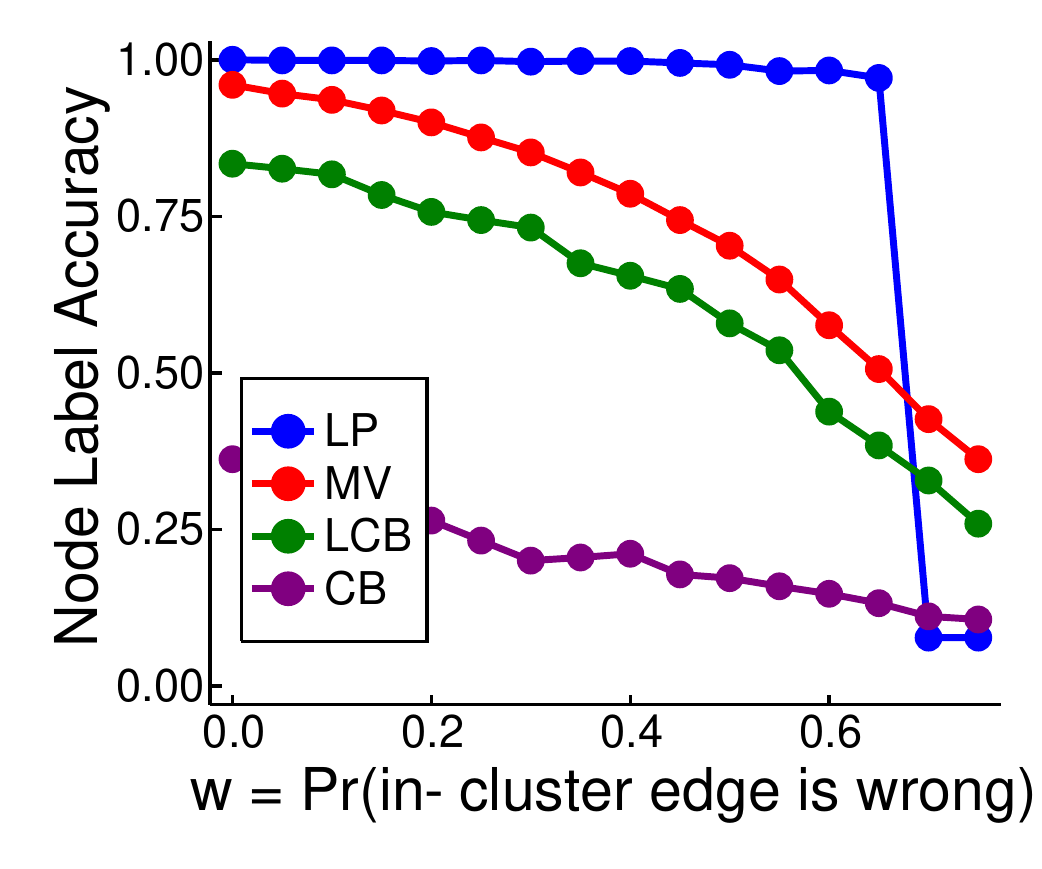}}\hfill
	\subfloat[Hypergraphs: Varying \# Clusters \label{fig:clusters-hyper}]
	{\includegraphics[width=.475\linewidth]{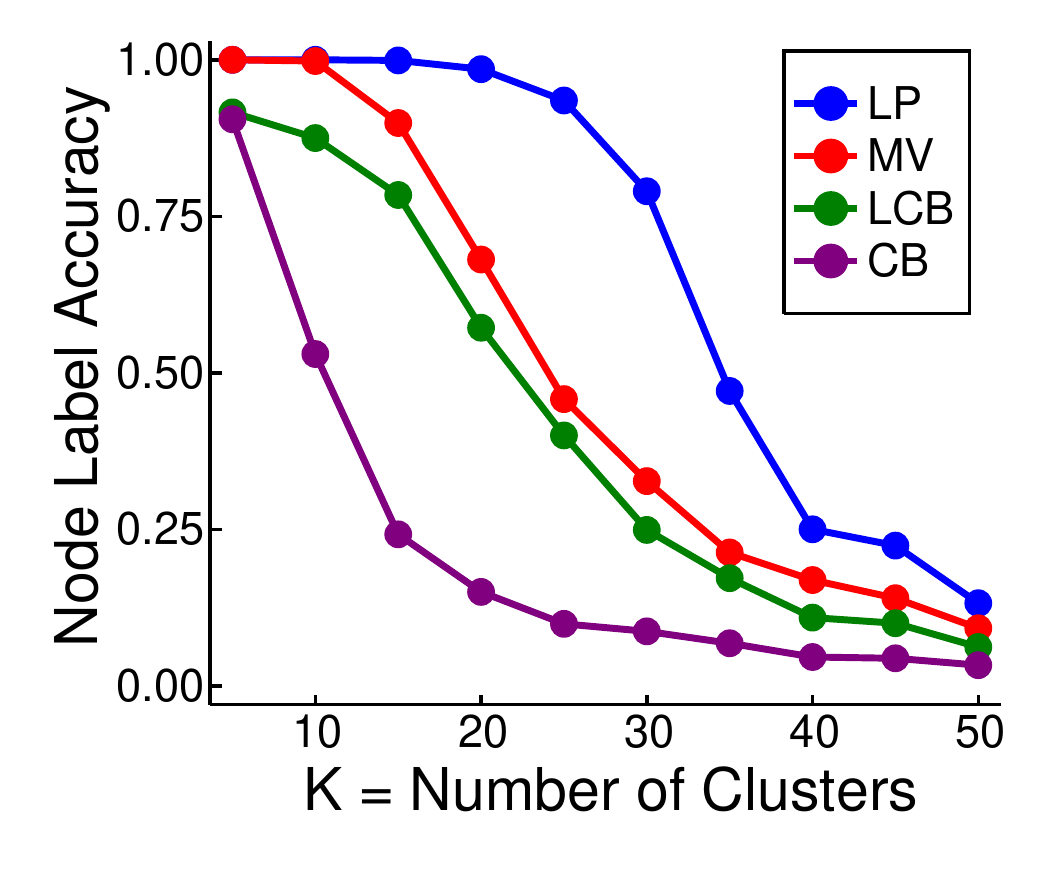}}
	\caption{(a)--(b): Performance of algorithms on a synthetic graph model for chromatic correlation clustering~\cite{Bonchi2015ccc}.
          Across a range of parameters, our \emph{LP} method outperforms competing methods in predicting the ground truth label of the nodes.
          (c)--(d): In experiments on synthetic 3-uniform hypergraphs, \emph{LP} performs well for most parameter regimes but there is some sensitivity to the very noisy setting.}
	\label{fig:synthetic_all}
\end{figure}

The CB and LCB algorithms, as well as the synthetic graph model itself, explicitly distinguish between ground truth node labels and ground truth clusters. Thus, our third experiment explores a parameter regime tailored more towards the strengths of CB and LCB. We fix $L = 20$, and vary the number of clusters from $K = 50$ to $K = 200$ in increments of 25. Following the experiments of Bonchi et al.~\cite{Bonchi2015ccc} we set $p = w = 0.5$, and set $q = 0.03$. Even in this setting, we find that our algorithms maintain an advantage. For all values of $K$, our LP algorithm outperforms other methods in terms of node label accuracy, and also obtains higher ARI scores when $K$ is a small multiple of $L$. We note that LCB and CB only obtain better cluster identification scores in parameter regimes where all algorithms obtain ARI scores below 0.1.

\begin{figure}[t]
	\centering
	\subfloat[Label Accuracy \label{acc}]
	{\includegraphics[width=.475\linewidth]{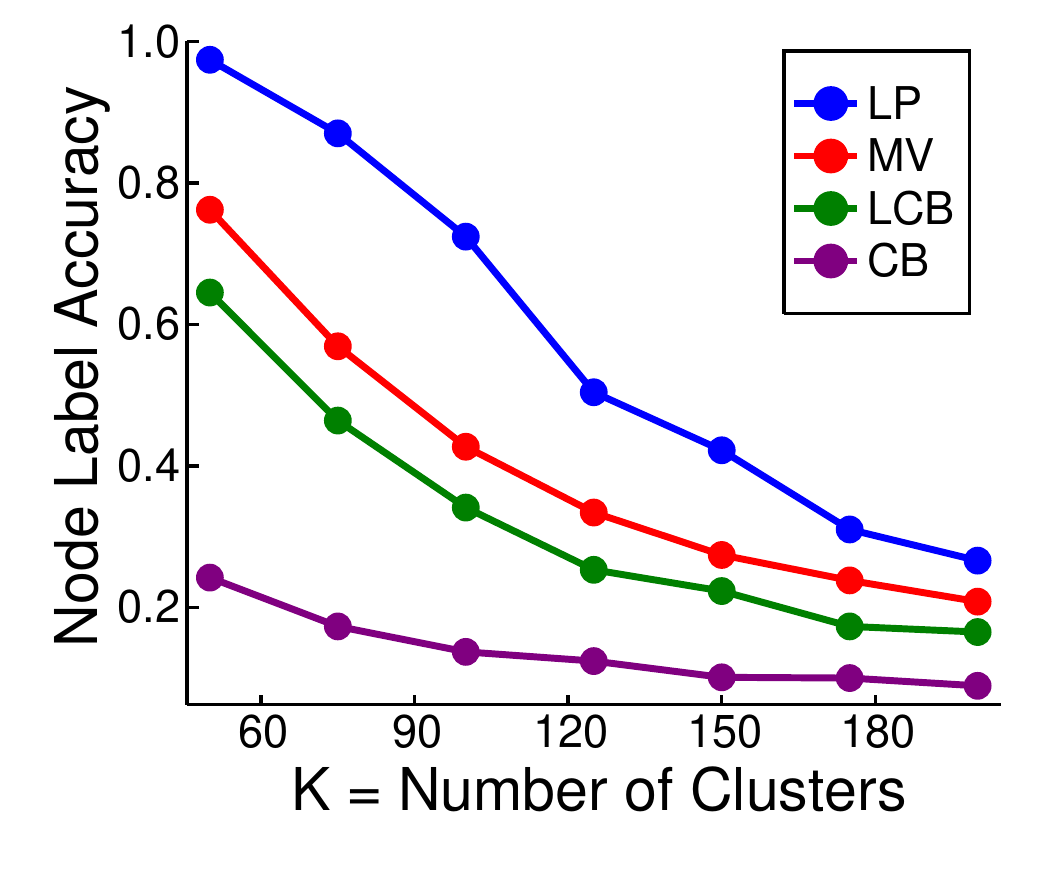}}\hfill
	\subfloat[ARI Scores \label{ari}]
	{\includegraphics[width=.475\linewidth]{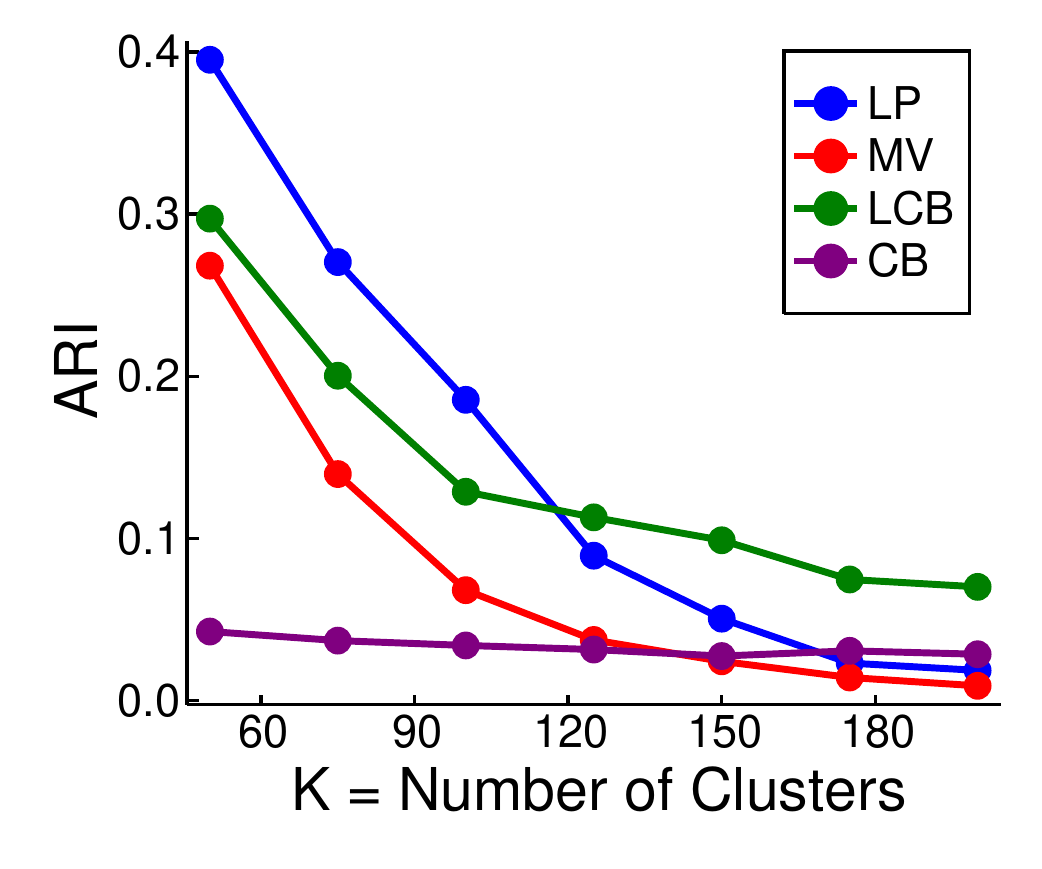}}
	\caption{LCB and CB are primarily designed for settings where $K$ is much larger than $L$. Despite this, our LP method always obtains better label assignment scores, and often obtains better ARI cluster identification scores, when we fix $L = 20$ and let $K$ vary from 50 to 500.}
	\label{fig:bonchiexps}
\end{figure}

\xhdr{Synthetic Hypergraph Results}
We ran similar experiments on synthetic 3-uniform hypergraphs. 
We again set $n = 1000$ and used $p = 0.005$ and $q = 0.0001$ for intra-cluster and inter-cluster hyperedge probabilities.
In one experiment, we fixed $L = 15$ and varied $w$, and in another we fixed $w = 0.2$ and varied the number of clusters $L$.
\Cref{fig:noise-hyper,fig:clusters-hyper} shows the accuracies.
Again, \emph{LP} tends to have the best performance. When $L = 15$, our method achieves nearly perfect accuracy for $w \le 0.6$.
However, we observe performance sensitivity when the noise is too large:
when $w$ increases from $0.6$ to $0.65$, the output of \emph{LP} no longer tracks the ground truth cluster assignment.
This occurs despite the fact that the LP solution is integral, and we are in fact optimally solving the \obj{} objective.
We conjecture this sharp change in accuracy is due to an information theoretic detectability threshold, which depends on parameters of the synthetic model.

\xhdr{Academic Department Labels in an Email Network} 
To test the algorithms on real-world data, we use the \emph{Email-Eu-core} network~\cite{yin2017local,leskovec2007graph}.
Nodes in the graph represent researchers at a European institution, 
edges indicate email correspondence (we consider the edges as undirected), 
and nodes are labeled by the departmental affiliation of each researcher. 
We wish to test how well each method can identify node labels, 
if we assume we have access to a (perhaps noisy and imperfect) mechanism for associating emails with labels for inter-department and intra-department communication. 
To model such a mechanism, we generate edge categories in a manner similar to the synthetic above.
An edge inside of a cluster (i.e., an email within the same department) is given the correct department label with probability $1-w$, and a random label with probability $w$. 
An edge between two members of different departments is given a uniform random label. 
\Cref{fig:email2} reports each algorithm's ability to detect department labels when $w$ varies from $0$ to $0.75$. Our \emph{LP}
method returns the best results in all cases, and is robust in detecting department labels even in the high-noise regime.
\begin{figure}[t]
	\centering
	\subfloat[Email-Eu-core \label{fig:email2}]
	{\includegraphics[width=.5\linewidth]{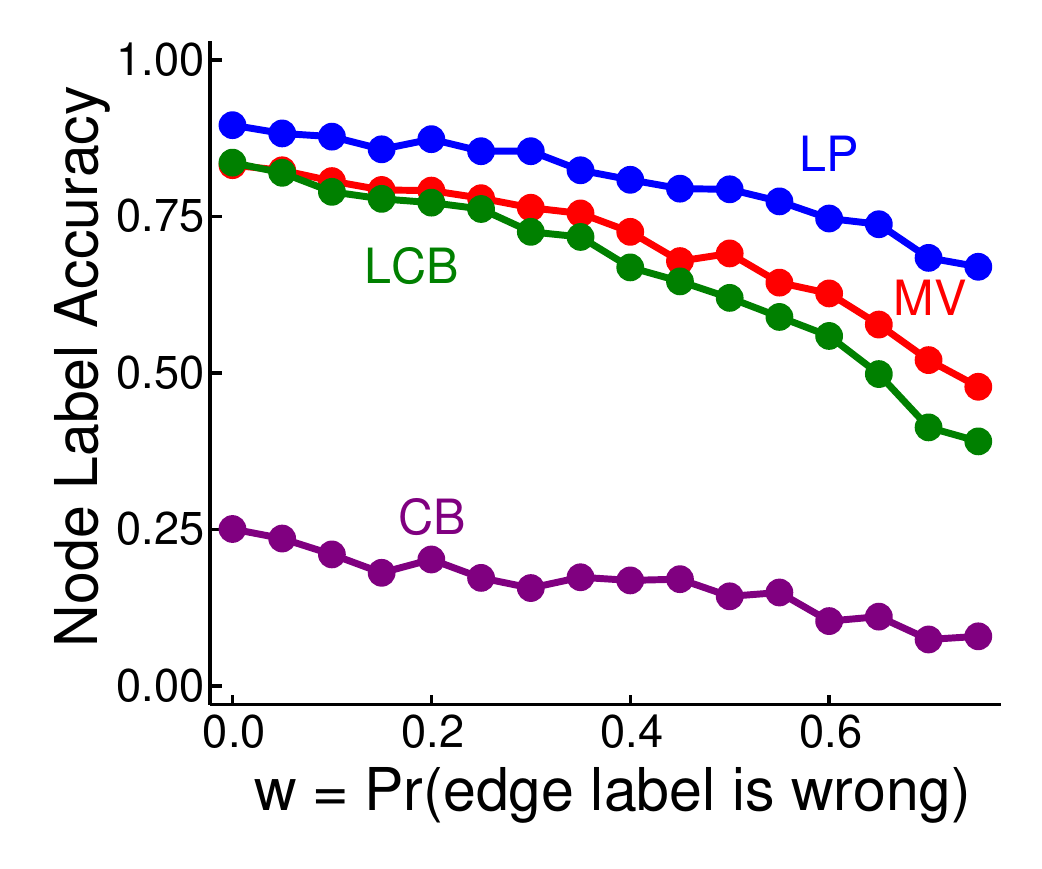}}\hfill
	\subfloat[Walmart-Products \label{fig:walmart}]
	{\includegraphics[width=.5\linewidth]{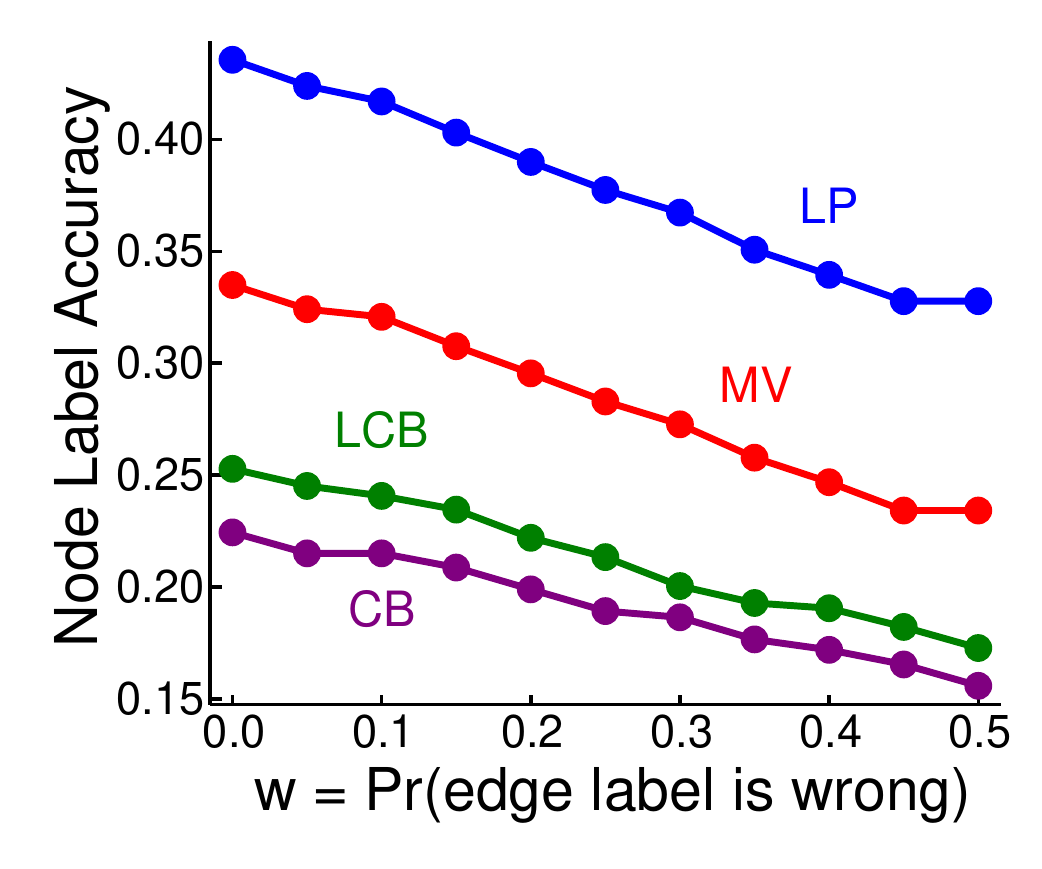}}\hfill
	\caption{
	Accuracy in clustering nodes in real-world datasets when edge labels are a
	noisy signal for ground truth node cluster membership.
	For both an email graph (a) and a product co-purchasing hypergraph (b), our \emph{LP-Round} method
	consistently outperforms other methods.
	}
	\label{fig:real-cd}
\end{figure}

\xhdr{Product Categories}
The \emph{Walmart-Trips} dataset from \cref{sec:exp1} also has product information.
We assigned products to one of ten broad departments in which they appear on \url{walmart.com}
(e.g., ``Clothing, Shoes, and Accessories'') to construct a \emph{Walmart-Products} hypergraph 
with ground truth node labels.
Recall that hyperedges are sets of co-purchased products. 
We generate noisy hyperedge labels as before, with $1-w$ as the probability that a hyperedge with nodes from a single department will have the correct label. 
Results are reported in \cref{fig:walmart}, and our
LP-round method can detect true departments at a much higher rate than the other methods.

\subsection{Temporal Community Detection}
In the next experiment, we show how our framework can be used to identify
communities of nodes in a temporal network, where we use timestamps on edges as
a type of categorical label that two nodes should be clustered together. For
data, we use the \emph{CollegeMsg} network~\cite{panzarasa2009patterns}, which
records private messages (time-stamped edges) between 1899 users (nodes) of a
social media platform at UC-Irvine.

Removing timestamps and applying a standard graph clustering algorithm would be
a standard approach to identify communities of users.
However, this loses the explicit relationship with time.
As an alternative, we convert timestamps into discrete edge labels
by ordering edges with respect to time and separating them into $k$ equal-sized
bins representing time windows. Optimizing \obj{} then corresponds to
clustering users into time windows, in order to maximize the number of private
messages that occur between users in the same time window. In this way, our
framework can identify \emph{temporal communities} in a social network, i.e.,
groups of users that are highly active in sending each other messages
\emph{within a short period of time}.
\begin{figure}[t]
	\centering
	\subfloat[Normalized Cut \label{fig:normcut}]
	{\includegraphics[width=.475\linewidth]{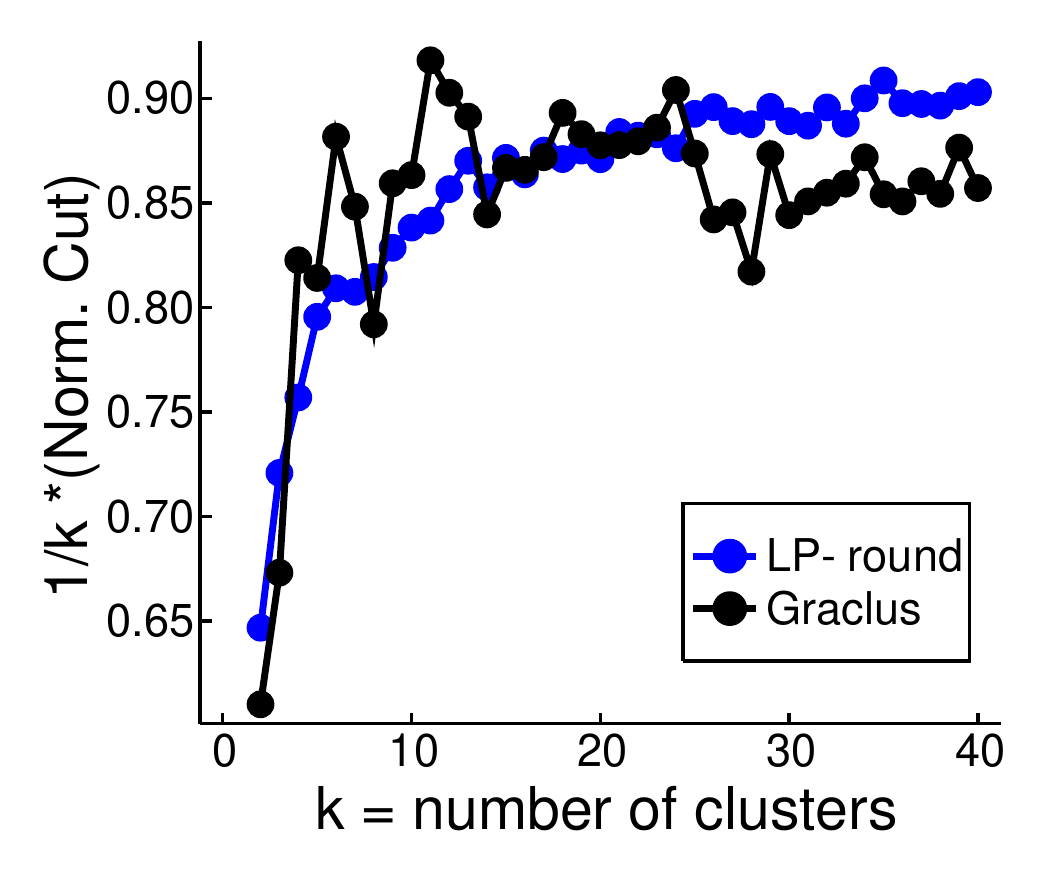}}\hfill
	\subfloat[Inner edge time difference \label{fig:avgtimediff}]
	{\includegraphics[width=.475\linewidth]{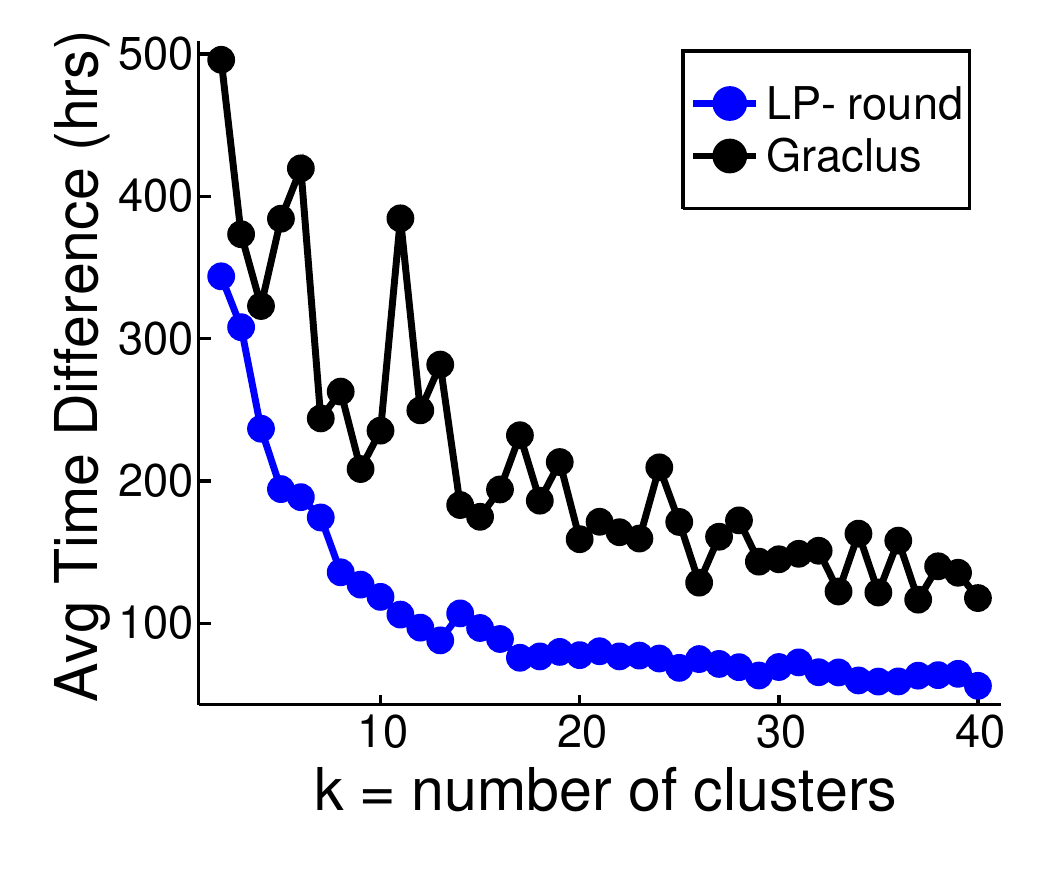}}
	\caption{Results for \emph{LP} and \emph{Graclus} in clustering a temporal network.
          Our \emph{LP} method is competitive for \emph{Graclus}'s objective (normalized cut; left),
          while preserving the temporal structure of network much better (right).}
	\label{fig:temporal}
\end{figure}

We construct edge-labeled graphs for different values of $k$, and compare
\emph{LP} against clusterings obtained by discarding time stamps and running \emph{Graclus}~\cite{Dhillon-2007-graclus},
a standard graph clustering algorithm.
\emph{Graclus} seeks to cluster the nodes into $k$ disjoint clusters
$S_1, \ldots, S_k$ to minimize the normalized cut objective:
\begin{equation*}
\label{normcut}
{\textstyle \textbf{Ncut}(S_1, S_2, \hdots, S_k) =\sum_{i = 1}^k \frac{\textbf{cut}(S_i)}{\textbf{vol}(S_i)}},
\end{equation*}
where $\textbf{cut}(S)$ is the number of edges leaving $S$, and
$\textbf{vol}(S)$ is the \emph{volume} of $S$, i.e., the number of edge end
points in $S$. \Cref{fig:normcut} shows that \emph{LP} is in fact
competitive with \emph{Graclus} in finding clusterings with small normalized cut
scores, even though \emph{LP} is designed for a different objective.
However, \emph{LP} still avoids cutting edges, and it finds clusterings that also have small normalized cut values.
The other goal of \emph{LP} is to place few edges in a cluster
with the wrong label, which in this scenario corresponds to clustering
messages together if they were sent close in time. We therefore also
measure the average difference between timestamps of interior edges and the
average time stamp in each clustering, i.e.,
\begin{equation*}
\label{timediff}
\textstyle {\textbf{AvgTimeDiff}(S_1, \hdots, S_k) = \frac{1}{\lvert E_{\textnormal{int}} \rvert}\sum_{i = 1}^k \sum_{e\in E_i} \lvert \text{timestamp}(e) - \mu_i \rvert, }
\end{equation*}
where $E_{\textnormal{int}}$ is the set of interior edges completely contained
in some cluster, $E_i$ is the set of interior edges of cluster $S_i$, and
$\mu_i$ is the average time stamp in $E_i$. Not surprisingly, this value tends
to be large for \emph{Graclus}, since this method ignores timestamps.
However, \Cref{fig:avgtimediff} shows that this value tends to be small for \emph{LP}, indicting that it
is indeed detecting clusters of users that are highly interactive within a
specific short period of time.

\subsection{Analysis of the Cooking Hypergraph}
\label{sec:exp3}
Finally, we apply our framework and LP-round algorithm to gain insights into the
\emph{Cooking} hypergraph dataset from \cref{sec:exp1}, demonstrating
our methodology for exploratory data analysis.
An edge in this hypergraph is a set of ingredients for a recipe, and each
recipe is categorized according to cuisine. \obj{}
thus corresponds to separating ingredients among cuisines, in a way that
maximizes the number of recipes whose ingredients are all in the same cluster
(see Ahn et al.~\cite{Ahn-2011-flavor} for related analyses).

\Cref{tab:allexp1} shows that only 20\% of the recipes can
be made (i.e., a 0.2 edge satisfaction) after partitioning ingredients among
cuisine types.  This is due to the large number of common ingredients such as
salt and olive oil that are shared across many cuisines (a problem in
other recipe network analyses~\cite{Teng-2012-recipes}). We negate
the negative effect of high-degree nodes as follows.  For an ingredient $i$, let
$d_i^c$ be the number of recipes of cuisine $c$ containing $i$. Let
$M_i = \max_{c} d_i^c$ measure \emph{majority degree} and $T_i = \sum_{t} d_i^c$ the
\emph{total degree}. Note that $B_i = T_i - M_i$ is a lower bound on the number
of hyperedge mistakes we will make at edges incident to node $i$. We can refine
the original dataset by removing all nodes with $B_i$ greater than some $\beta$.

\xhdr{Making recipes or wasting ingredients}
\begin{figure}[t]
	\centering
	\subfloat[Edge Satisfaction \label{fig:esat}]
	{\includegraphics[width=.45\linewidth]{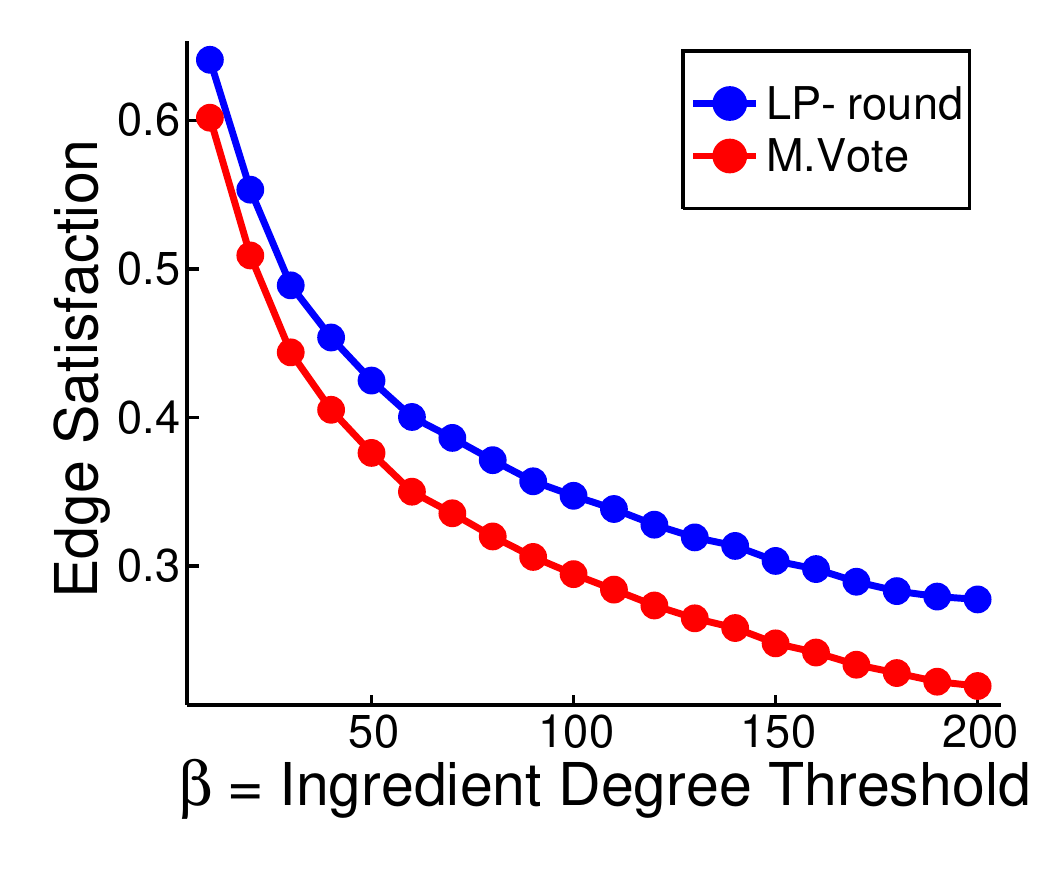}}\hfill
	\subfloat[Unused ingredients \label{fig:unused}]
	{\includegraphics[width=.475\linewidth]{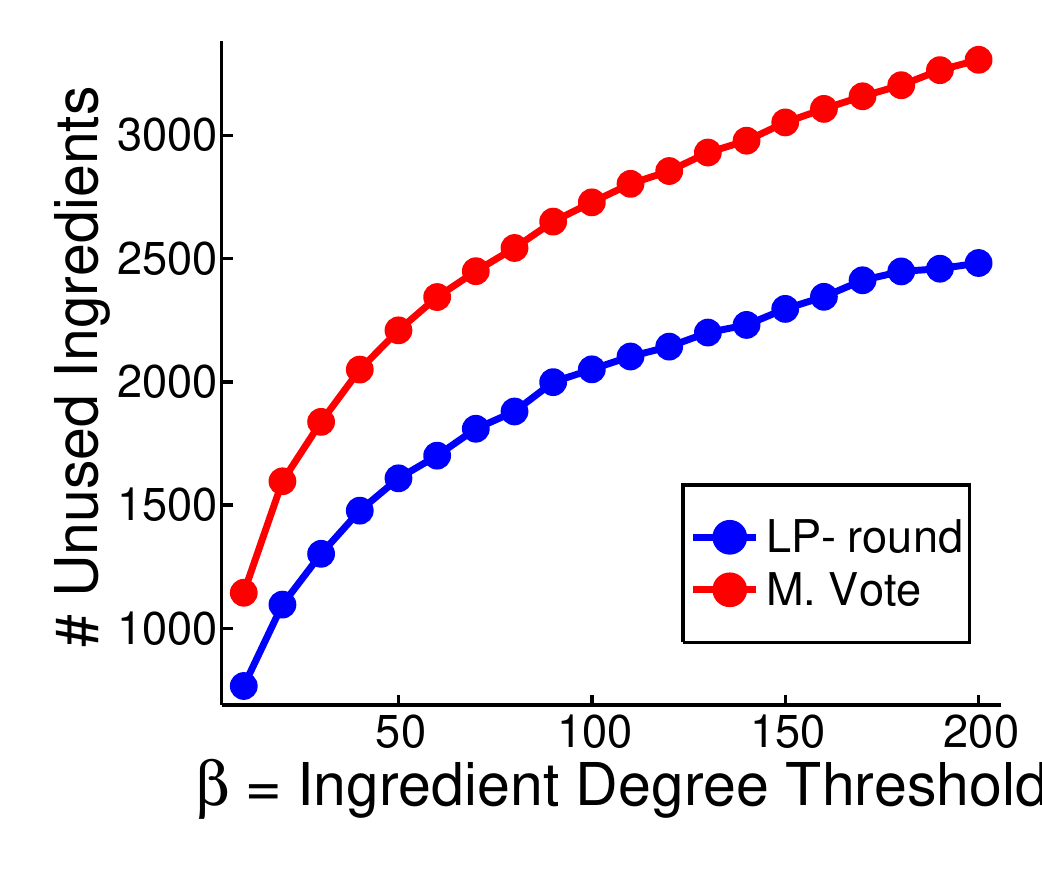}}\hfill
	\caption{As $\beta$ increases, we discard fewer high-degree ingredients before clustering the rest. Our method
	 always ``makes'' more recipes (higher edge satisfaction) and ``wastes'' fewer ingredients (smaller number of unused ingredients).}
	\label{fig:cooking}
\end{figure}
\Cref{fig:esat} shows edge satisfaction scores for \emph{LP} and \emph{MV} when
we cluster for different $\beta$. When $\beta = 10$, edge satisfaction is above
0.64 with \emph{LP}. As $\beta$ increases, edge satisfaction decreases, but
\emph{LP} outperforms \emph{MV} in all cases. We also consider a measure of
``ingredient waste'' for each method. An ingredient is \emph{unused} if we
cannot make any recipes by combining the ingredient with other ingredients in
its cluster. A low number of unused ingredients indicates that a method forms
clusters where ingredients combine together well. \Cref{fig:unused}
shows the number of unused ingredients as $\beta$ varies. Again, \emph{LP}
outperforms \emph{MV}.

\xhdr{Specific ingredient and recipe clusters} 
\begin{table}[t]
	\caption{Examples of ingredients and recipes from special clusters identified by \emph{LP}, but not \emph{Majority Vote}.}
	\label{tab:clusters}
	\centering
	\begin{tabular}{p{\linewidth}}
		\toprule
		\textbf{French Fruit-Based Desserts} ($\beta = 70$)\\ 
		\textbf{Ingredients:} ruby red grapefruit, strawberry ice cream, dry hard cider, icing, prunes, tangerine juice, sour cherries.\\
		 \textbf{Recipes:} 1. \{almond extract, bittersweet chocolate, sugar, sour cherries, brioche, heavy cream, unsalted butter, kirsch\}, 2. \{large egg yolks, ruby red grapefruit, dessert wine, sugar\}\\
		\midrule 
		\textbf{Brazilian Caipirinha Recipes} ($\beta = 170$) \\
		\textbf{Ingredients:} simple syrup, light rum, ice, superfine sugar, key lime, coco, kumquats, liquor, mango nectar, vanilla essence\\
		 \textbf{Recipes:} \{cachaca, ice\} $+$ 
		1. \{ lime juice, kumquats, sugar\},
		2. \{lime, fruit puree, simple syrup\},
		3. \{ superfine sugar, lime juice, passion fruit juice\},
		4. \{ sugar, liquor, mango nectar, lime, mango\}\\
		\bottomrule
	\end{tabular}
\end{table} 
We finally highlight specific ingredient clusters that \emph{LP} identifies but \emph{MV} does not. When $\beta = 170$, \emph{LP} places 10 ingredients with the Brazilian cuisine which \emph{MV} does not, leading to 23 extra recipes that are unique to \emph{LP}. Of these, 21 correspond to variants of the Caipirinha, a popular Brazilian cocktail. When $\beta = 70$, 24 ingredients and 24 recipes are unique to the French cuisine cluster of \emph{LP}. Of these, 18 correspond to desserts, and 14 have a significant fruit component. 
\Cref{tab:clusters} has examples of ingredients and recipes from both these clusters.

\section{Discussion}
We have developed a computational framework for clustering nodes of 
hypergraphs when edges have categorical labels that signal node similarity.
With two categories, our clustering objective can be solved in polynomial time.
For general problems, our linear programming relaxations provide
2-approximation or even better guarantees, which are far tighter than what is
seen in the related literature on correlation clustering. This method is also
extremely effective in practice. Amazingly, our LP-round algorithm often actually minimizes
our NP-hard objective (certified through integral solutions) on hypergraphs with
tens of thousands of edges in just tens of seconds. The approach also works well
in problems when performance is measured in terms of some sort of ground truth labeling, outperforming baselines by a
substantial margin.

For the special cases of two-category graphs and rank-$3$ hypergraphs, the
\obj{} objective is a ``regular energy function'' within the energy minimization
framework of computer vision~\cite{kolmogorov2004what}. This provides
alternative polynomial time algorithms in these cases (see
\cref{sec:energy}). However, these approaches do not work for two important
regimes: more than two categories, or in general hypergraphs (in the latter, the
penalties are no longer a semi-metric, which is needed for approximation
algorithms~\cite{boykov2001fast}).

\section*{Acknowledgments}
This research was supported by
NSF Award DMS-1830274,
ARO Award W911NF19-1-0057, and
ARO MURI.

\bibliographystyle{plain}
\bibliography{main}

\begin{appendix}
\section{Connection to energy minimization}\label{sec:energy}
Special cases of our \obj{} framework fit within the paradigm of energy function minimization in computer vision~\cite{boykov2001fast,kolmogorov2004what,freedman2005energy}.
The energy minimization approach uses minimum $s$-$t$ cut algorithms for functions of binary and ternary variables which satisfy a certain regularity property.
In this appendix we show that our objective induces a regular energy function in both the graph and
rank-3 hypergraph case when there are two categories.
This connection implies that in addition to the algorithms we developed above,
we may use the tools developed for energy minimization to facilitate solving
these special instances of our problem exactly in polynomial time.

\subsection{Graphs with two categories}\label{sec:graph}
To show the connection to energy minimization, we cast our objective as a so-called energy function.
With two categories, we can encode our coloring $\mathcal{C}$ of the $n$ nodes in the graph as a vector of 1s and 2s corresponding to which color each node takes.
For this, we write $\mathcal{C}=(x_1,....,x_n)$ where $x_i=1$ if node $i$ is assigned color 1 and $x_i=2$ if node $i$ is assigned a color 2.
Now the \obj{} objective can be written as an energy function as
\[
\cc(\mathcal{C})=\sum_{i<j\mbox{ s.t. } (i,j)\in E}E^{i,j}(x_i,x_j),
\]
where
\[
\begin{bmatrix}E^{i,j}(1,1) & E^{i,j}(1,2)\\E^{i,j}(2,1) & E^{i,j}(2,2)\end{bmatrix}=\begin{bmatrix} 0 & 1\\1 & 1\end{bmatrix}
\]
if $(i,j)$ is of color $1$ and
\[
\begin{bmatrix}E^{i,j}(1,1) & E^{i,j}(1,2)\\E^{i,j}(2,1) & E^{i,j}(2,2)\end{bmatrix}=\begin{bmatrix} 1 & 1\\1 & 0\end{bmatrix}
\]
if $(i,j)$ is of color $2$.
We will show that this energy function satisfies a regularity property, which enables a reduction of our objective to a minimum $s$-$t$ graph cut~\cite{kolmogorov2004what}.

\begin{definition}\label{def:regular}
A function of two binary variables is \emph{regular} if each term satisfies the following inequality
$$E^{i,j}(0,0)+E^{i,j}(1,1)\leq E^{i,j}(0,1)+E^{i,j}(1,0).$$
\end{definition}

It is easy to see that our energy function is indeed regular. We formalize this observation in the following theorem.

\begin{theorem}\label{thm:regular}
The \obj{} objective for graphs with two categories induces a regular energy function.
\end{theorem}
\begin{proof} Regardless of whether $(i,j)$ is an edge of color 1 or of color 2, the off-diagonal terms in the energy function sum to 2 while the diagonal terms sum to 1.
  This ensures that the regularity property is satisfied.
\end{proof}

Having established the regularity of our energy function,
the results of Kolmogorov and Zabih~\cite[Theorem 4.1]{kolmogorov2004what} say that we can cast the energy minimization problem as an $s$-$t$-cut problem on a directed graph.
In particular, following their construction, we would create a directed graph $G'=(V',E')$ from $G=(V,E,C,\ell)$ as follows,
which is similar to the reduction we used in \cref{sec:graph2}.
\begin{itemize}
    \item Append nodes $s$ and $t$ to $E'$
    \item For every undirected edge $(i,j)$ with $i<j$ in $G$, if $(i,j)$ has color 1, create a directed edge $(i,j)$ and a directed edge $(j,t)$ in $E'$, while if $(i,j)$ has color 2 in $G$, append the directed edge $(i,j)$ and the directed edge $(s,i)$ to $E'$
\end{itemize}
This construction guarantees that the energy function $E^{i,j}$ of every edge $(i,j)$ in $G$ is exactly represented by the corresponding cut on the subgraph in $G'$ which the edge induced. 
The following theorem is then a result of the additivity theorem from Kolmogorov and Zabih~\cite{kolmogorov2004what}.
\begin{theorem}
Let $\mathcal{C^*}$ be a two-colored clustering that is the solution of the $s$-$t$-mincut problem on the graph $G'$ constructed using the procedure above. Then $\mathcal{C^*}$ also optimizes the \obj{} objective for the original graph $G$.
\end{theorem}

\subsection{Rank-3 hypergraphs with two categories}
The energy minimization framework also allows us to handle the case of rank-3 hypergraphs.
Adopting the conventions of the previous subsection, we can write the clustering objective as follows. 
\begin{align*}
  & \cc(\mathcal{C}) = \\
  & \sum_{i<j \;:\; (i,j)\in E}E^{i,j}(x_i,x_j) + \sum_{i<j<k \;:\; (i,j,k)\in E}E^{i,j,k}(x_i,x_j,x_k),
\end{align*}
where $E^{i,j}(x_i,x_j)$ is the same as in the previous section and the higher-order energy for hyperedges is
\begin{align*}
  &\begin{bmatrix}E^{i,j,k}(1,1,1) & E^{i,j,k}(1,1,2)\\E^{i,j,k}(1,2,1) & E^{i,j,k}(1,2,2)\\E^{i,j,k}(2,1,1) & E^{i,j,k}(2,1,2)\\E^{i,j,k}(2,2,1) & E^{i,j,k}(2,2,2)\end{bmatrix} \\
  &\qquad =\begin{bmatrix}0&1\\1&1\\1&1\\1&0\end{bmatrix} +
  \mathbb{I}(C[i]=C[j]=C[k]=2)\begin{bmatrix}1&0\\0&0\\0&0\\0&0\end{bmatrix}  \\
       &\qquad\qquad +\mathbb{I}(C[i]=C[j]=C[k]=1)\begin{bmatrix}0&0\\0&0\\0&0\\0&1\end{bmatrix}.
\end{align*}
The energy function defined this way is regular, in the sense that all projections into two variables are regular.
We formalize this observation in the theorem below.
\begin{theorem}
The \obj{} objective for rank-3 hypergraphs with two categories induces a regular energy function.
\end{theorem}

We proceed to construct a graph $G'$ in a manner similar to that described in
the preceding subsection which will allow us to optimize the \obj{} objective
through a minimum $s$-$t$-cut on $G'$. After appending the source and sink nodes
$s$ and $t$ to $G'$, we perform the procedure of the previous section for all
edges $e\in E$. For the remaining hyperedges of rank 3, we follow the procedure
outlined by Kolmogorov and Zabih~\cite[Section 4.1]{kolmogorov2004what}.
This is a special case of the more general approach we present in
\cref{sec:twocats}. In particular, depending on the hyperedge
color, we use one of the two directed tree structures in \cref{fig:subgraphs}.
The fact that the minimum $s$-$t$ cut on $G'$ thus constructed induces a partition of the nodes in $E$
which minimizes the \obj{} objective follows from a proof similar to that
presented in the graph case. The actual proof is the special $r=3$ case of the
main proof in \cref{sec:twocats}.
Finally, we can establish the following theorem.
\begin{theorem}
Let $\mathcal{C^*}$ be a two-colored clustering that is the solution of the $s$-$t$-mincut problem on the graph $G'$ constructed using the procedure above. Then $\mathcal{C^*}$ also optimizes the \obj{} objective for the original hypergraph $G$.
\end{theorem}

\end{appendix}


\end{document}